\documentclass[envcountsame]{llncs}

\usepackage{times}
\usepackage{microtype}

\usepackage{url}

\usepackage{amssymb,amsmath}
\usepackage{pifont}
\usepackage{tikz}
\usetikzlibrary{positioning,arrows,fit,calc}
\usetikzlibrary{decorations.pathmorphing}
\usetikzlibrary{shapes,patterns}

\newcommand{\q}{{\boldsymbol{q}}}
\newcommand{\T}{{\cal T}}
\newcommand{\A}{{\cal A}}
\newcommand{\C}{{\cal C}}

\newcommand{\R}{\mathsf{sub}(T)}
\newcommand{\can}{\C_{\T, \A}}
\newcommand{\ind}{\mathsf{ind}}
\newcommand{\dom}{\mathsf{dom}}
\newcommand{\degree}{\mathsf{deg}}
\newcommand{\restr}{\upharpoonright}
\newcommand{\dD}{\partial D}
\newcommand{\dDp}{\partial D'}
\newcommand{\dDpp}{\partial D''}
\newcommand{\rpred}{G}
\def\t{\mathsf{t}}
\newcommand{\ti}{\mathsf{t}_\mathsf{i}}
\newcommand{\tr}{\mathsf{t}_\mathsf{r}}
\newcommand{\rew}{P}

\newcommand{\omq}{{\ensuremath{\boldsymbol{Q}}}}

\newcommand{\cir}{\boldsymbol{C}}
\newcommand{\wid}{\mathsf{w}}
\newcommand{\dep}{\mathsf{d}}
\newcommand{\ari}{\mathsf{r}}

\newcommand{\LOGCFL}{\ensuremath{\mathsf{LOGCFL}}}

\newcommand{\OWLQL}{\textsl{OWL\,2\,QL}}

\newcommand{\FO}{\text{FO}}

\newcommand{\NDL}{\text{NDL}}

\newcommand{\NCone}{{\ensuremath{\mathsf{NC}^1}}}

\newcommand{\NP}{\ensuremath{\mathsf{NP}}}
\newcommand{\PTime}{\mathsf{P}}

\newcommand{\NLpoly}{\ensuremath{\mathsf{NL}}/\ensuremath{\mathsf{poly}}}
\newcommand{\NL}{\ensuremath{\mathsf{NL}}}

\newcommand{\vars}{\mathsf{var}}
\newcommand{\avars}{\mathsf{avar}}

\newcommand{\gfmn}{\mathcal{G}}
\renewcommand{\qed}{\hfill\ding{113}}
\newcommand{\exs}{{\scriptscriptstyle\exists}}

\tikzset{bpoint/.style={circle,inner sep=0pt,minimum size=1.5mm,fill=black,draw=black},
	wpoint/.style={circle,inner sep=0pt,minimum size=1.5mm,fill=white,draw=black,thick}}

\newcommand{\nd}{t}
\newcommand{\tpd}{\vec{w}}
\newcommand{\tpr}{\vec{s}}

\newcommand{\role}{\boldsymbol{R}}
\newcommand{\rni}{\ensuremath{\role_\T}}

\newcommand{\twords}{\mathbf{W}_\T}
\newcommand{\sqset}{\mathsf{SQ}}

\newcommand{\nb}[1]{\hbox to 0pt{\textcolor{red}{\bf !}}\marginpar{\parbox{20mm}{\scriptsize\raggedright\textcolor{red}{#1}}}}

\title{Theoretically Optimal Datalog Rewritings for \OWLQL{} Ontology-Mediated Queries}
\author{M. Bienvenu$^1$\! \and S. Kikot$^2$\! \and R. Kontchakov$^2$\! \and V. Podolskii$^3$\! \and M. Zakharyaschev$^2$}
 \institute{
   $^1$ CNRS \& University of Montpellier, France   
   (\url{meghyn@lirmm.fr})\\
   $^2$ Birkbeck, University of London, U.K.\ 
   (\url{{kikot,roman,michael}@dcs.bbk.ac.uk})\\
   $^3$ Steklov Mathematical Institute, Moscow, Russia 
   (\url{podolskii@mi.ras.ru})}

\begin{document}
\maketitle

\begin{abstract}
We show that, for \OWLQL{} ontology-mediated queries with (\emph{i})  ontologies of bounded depth and conjunctive queries of bounded treewidth, (\emph{ii})~ontologies of bounded depth and bounded-leaf tree-shaped conjunctive queries, and (\emph{iii}) arbitrary ontologies and bounded-leaf tree-shaped conjunctive queries, one can construct and evaluate nonrecursive datalog rewritings by, respectively, \LOGCFL, \NL{} and \LOGCFL{} algorithms, which matches the optimal combined complexity.
\end{abstract}

\section{Introduction}

\begin{figure}[t]%
\tikzset{cmplx/.style={draw,thick,rounded corners,inner sep=0mm}}%
\mbox{}\hspace*{-0.3em}\begin{tikzpicture}[x=7.2mm,y=4.5mm]
\draw[thick] (0.4,-0.4) rectangle +(8.2,-6.2);
\node[rotate=90] at (-0.3,-4.0) {\scriptsize ontology depth};
\begin{scope}[ultra thin]
\draw (0.4,-1) -- +(8.2,0); \node at (0.1,-1) {\scriptsize 1};
\draw (0.4,-2) -- +(8.2,0); \node at (0.1,-2) {\scriptsize 2};
\draw (0.4,-3) -- +(8.2,0); \node at (0.1,-3) {\scriptsize 3};
\draw (0.4,-4) -- +(8.2,0); \node at (0,-4) {\scriptsize \dots};
\draw (0.4,-5) -- +(8.2,0); \node at (0.1,-5) {\scriptsize $d$};
\draw (0.4,-6) -- +(8.2,0); \node at (0,-6) {\scriptsize arb.};
\draw (1,-0.4) -- +(0,-6.2); \node at (1,-7) {\scriptsize 2}; 
\draw (2,-0.4) -- +(0,-6.2); \node at (2,-7) {\scriptsize \dots}; 
\draw (3,-0.4) -- +(0,-6.2); \node at (3,-7) {\scriptsize $\ell$}; 
\draw (4,-0.4) -- +(0,-6.2); \node at (4,-7) {\scriptsize trees}; 
\draw (5,-0.4) -- +(0,-6.2); \node at (5,-7) {\scriptsize 2}; 
\draw (6,-0.4) -- +(0,-6.2); \node at (6,-7) {\scriptsize \dots}; 
\draw (7,-0.4) -- +(0,-6.2); \node at (7,-7) {\scriptsize bound.}; 
\draw (8,-0.4) -- +(0,-6.2); \node at (8,-7) {\scriptsize arb.}; 
\end{scope}
\node at (2,-7.5) {\scriptsize number of leaves};
\node at (6.5,-7.5) {\scriptsize treewidth};
\node [fill=gray!25,cmplx,fill opacity=0.9,fit={(0.6,-1.6) (3.4,-6.4)}] 
{\raisebox{-6ex}{\begin{tabular}{c}poly NDL\\[2pt] no poly PE\\[2pt] 
\footnotesize poly FO\\[-3pt]\scriptsize iff\\[-3pt]\scriptsize \NLpoly  $\,\subseteq\,$ \NCone\end{tabular}}};
\node [fill=gray!50,cmplx,fill opacity=0.9,fit={(3.6,-1.6) (7.4,-5.4)}] 
{\raisebox{-7ex}{\begin{tabular}{c}poly NDL\\[1pt] no poly PE\\[1pt] 
\footnotesize poly FO\\[-5pt] \scriptsize iff\\[-5pt] \scriptsize LOGCFL/poly $\!\subseteq\!$ \NCone\end{tabular}}};
\node [fill=black,cmplx,fit={(3.6,-5.6) (8.4,-6.4)}] 
{\hspace*{-1.8em}\raisebox{-9pt}{\textcolor{white}{\begin{tabular}{c}\small no poly NDL\,{\scriptsize\&}\,PE
\end{tabular}}}}; 
\node [fill=black,cmplx,fit={(7.6,-1.6) (8.4,-6.4)}]  
{\rotatebox{90}{\hspace*{-3.5em}\textcolor{white}{\begin{tabular}{c}
\scriptsize \begin{tabular}{c}poly\\FO\end{tabular}\! iff  NP{\tiny\!/\!poly} $\!\!\subseteq\!\!$ NC$^1$
\end{tabular}}}};
\node [fill=gray!5,cmplx,fill opacity=0.9,fit={(0.6,-0.6) (4.4,-1.4)}]  {\raisebox{-1.5ex}{poly $\Pi_4$-PE}};
\node [fill=gray!5,cmplx,fill opacity=0.9,fit={(4.6,-0.6) (7.4,-1.4)}]  {\raisebox{-1.5ex}{poly PE}};
\node [fill=gray!25,cmplx,fill opacity=0.9,fit={(7.6,-0.6) (8.4,-1.4)}] { };
\node[inner sep=0pt] (test) at (5.5,0.3) {\begin{tabular}{c}poly NDL, {\scriptsize but} no poly PE\\[-3pt]\footnotesize poly FO \scriptsize iff \NLpoly  $\,\subseteq\,$ \NCone\end{tabular}};
\draw (8,-1) -- (test.-15);
\node at (-0.05,0.3) {\footnotesize (a)};
\end{tikzpicture}
\hspace*{0.2em}
\begin{tikzpicture}[x=5.8mm,y=4.5mm]
\draw[thick] (0.4,-0.4) rectangle +(8.2,-6.2);
\begin{scope}[ultra thin]
\draw (0.4,-1) -- +(8.2,0); \node at (0.1,-1) {\scriptsize 1};
\draw (0.4,-2) -- +(8.2,0); \node at (0.1,-2) {\scriptsize 2};
\draw (0.4,-3) -- +(8.2,0); \node at (0.1,-3) {\scriptsize 3};
\draw (0.4,-4) -- +(8.2,0); \node at (0,-4) {\scriptsize \dots};
\draw (0.4,-5) -- +(8.2,0); \node at (0.1,-5) {\scriptsize $d$};
\draw (0.4,-6) -- +(8.2,0); \node at (0,-6) {\scriptsize arb.};
\draw (1,-0.4) -- +(0,-6.2); \node at (1,-7) {\scriptsize 2}; 
\draw (2,-0.4) -- +(0,-6.2); \node at (2,-7) {\scriptsize \dots}; 
\draw (3,-0.4) -- +(0,-6.2); \node at (3,-7) {\scriptsize $\ell$}; 
\draw (4,-0.4) -- +(0,-6.2); \node at (4,-7) {\scriptsize trees}; 
\draw (5,-0.4) -- +(0,-6.2); \node at (5,-7) {\scriptsize 2}; 
\draw (6,-0.4) -- +(0,-6.2); \node at (6,-7) {\scriptsize \dots}; 
\draw (7,-0.4) -- +(0,-6.2); \node at (7,-7) {\scriptsize bound.}; 
\draw (8,-0.4) -- +(0,-6.2); \node at (8,-7) {\scriptsize\hspace*{0.5em} arb.}; 
\end{scope}
\node at (2,-7.5) {\scriptsize number of leaves};
\node at (6.5,-7.5) {\scriptsize treewidth};
\node [fill=gray!5,cmplx,fill opacity=0.9,fit={(3.4,-5.4) (0.6,-0.6)}]  {\NL};
\node [fill=gray!40,cmplx,fill opacity=0.9,fit={(3.6,-0.6) (7.4,-5.4)}]  {\LOGCFL};
\node [fill=black,cmplx,fit={(8.4,-6.4) (7.6,-0.6)}]  {}; 
\node [fill=black,cmplx,fit={(3.6,-5.6) (8.4,-6.4)}]  {\hspace*{7em}\raisebox{-1ex}{\textcolor{white}{\NP}}};
\node [fill=gray!40,cmplx,fill opacity=0.9,fit={(0.6,-5.6) (3.4,-6.4)}]  {\raisebox{-1ex}{\LOGCFL}};
\node at (-0.05,0.3) {\footnotesize (b)};
\end{tikzpicture}%
\caption{(a) Size of OMQ rewritings; (b) combined complexity of OMQ evaluation.}
\label{pic:results}
\end{figure}
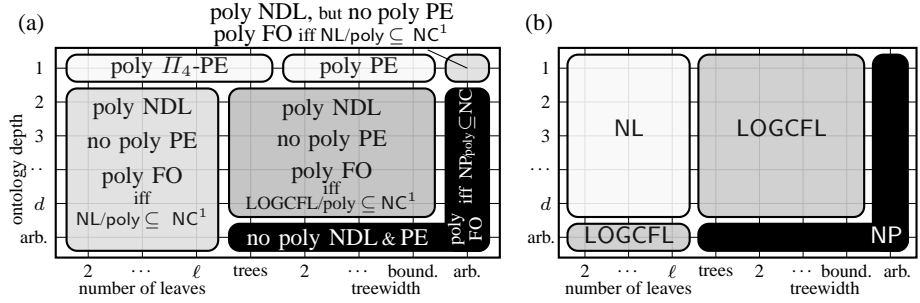

Ontology-based data access (OBDA) via query rewriting~\cite{PLCD*08} reduces the problem of finding answers to conjunctive queries (CQs) mediated by \OWLQL{} ontologies to standard database query answering. The question we are concerned with here is whether this reduction is optimal with respect to the combined complexity of query evaluation. Figure~\ref{pic:results}~(a) summarises what is known  about  the size of positive existential (PE), nonrecursive datalog (NDL) and first-order (FO) rewritings of \OWLQL{} ontology-mediated queries (OMQs) depending on the existential depth of their ontologies and the shape of their CQs~\cite{DBLP:conf/icalp/KikotKPZ12,DBLP:journals/ai/GottlobKKPSZ14,LICS14,DBLP:conf/lics/BienvenuKP15}. Figure~\ref{pic:results}~(b) shows the combined complexity of OMQ evaluation for the corresponding classes of OMQs~\cite{CDLLR07,DBLP:conf/dlog/KikotKZ11,LICS14,DBLP:conf/lics/BienvenuKP15}. Thus, we see, for example, that PE-rewritings for OMQs with ontologies of bounded depth and CQs of bounded treewidth can be of super-polynomial size, and so not evaluable in polynomial time, while the evaluation problem for these OMQs is decidable in $\LOGCFL \subseteq \PTime$. On the other hand, the OMQs in this class enjoy polynomial-size NDL-rewritings. However, these rewritings were defined using an argument from circuit complexity~\cite{DBLP:conf/lics/BienvenuKP15}, and it has been unclear whether they can be constructed and evaluated in \LOGCFL. The same concerns the class of OMQs with ontologies of bounded depth and bounded-leaf tree-shaped queries, which can be evaluated in \NL, and the class of OMQs  with arbitrary ontologies and bounded-leaf tree-shaped queries, which can be evaluated in \LOGCFL. 

In this paper, we consider OMQs in these three classes and construct NDL-rewritings that are theoretically optimal in the sense that the rewriting and evaluation can be carried out by algorithms of optimal combined complexity, that is, from the complexity classes \LOGCFL, \NL{} and \LOGCFL, respectively. Such algorithms are known to be space efficient and highly parallelisable.
We compared our optimal NDL rewritings with those produced by query rewriting engines Clipper~\cite{DBLP:conf/aaai/EiterOSTX12} and Rapid~\cite{DBLP:conf/cade/ChortarasTS11}, using a sequence of OMQs with linear CQs and a fixed ontology of depth 1. 



\section{Preliminaries}\label{sec:prelims}

We give $\OWLQL$ in the DL syntax with \emph{individual names} $a_i$, \emph{concept names}~$A_i$, and \emph{role names} $P_i$ ($i\ge 1$). \emph{Roles} $R$ and \emph{basic concepts} $B$ are defined by
\begin{equation*}
R \quad ::=\quad P_i \quad\mid\quad P_i^-, \qquad\qquad
B \quad ::=\quad 
A_i \quad\mid\quad \exists R.
\end{equation*}
A \emph{TBox}, $\mathcal{T}$, is a finite set of inclusions of the form 
\begin{equation*}
B_1 \sqsubseteq B_2, \qquad B_1 \sqcap B_2 \sqsubseteq \bot, \qquad  R_1 \sqsubseteq R_2,\qquad R_1 \sqcap R_2 \sqsubseteq \bot. \end{equation*}
An \emph{ABox}, $\A$, is a finite set of atoms of
the form $A_k(a_i)$ or $P_k(a_i,a_j)$. We denote by $\ind(\A)$ the set of individual names in $\A$, and by $\role_\T$ the set of role names occurring in $\T$ and their inverses. We use $A \equiv B$ for $A\sqsubseteq B$ and $B \sqsubseteq A$.
The semantics for \OWLQL{} is defined in the usual way based on interpretations $\mathcal{I} = (\Delta^\mathcal{I}, \cdot^\mathcal{I})$~\cite{BCMNP03}. 

For every role $R\in\role_\T$, we take a fresh concept name $A_R$ and add $A_R \equiv \exists R$ to $\T$. The resulting TBox is said to be in \emph{normal form}, and we assume, without loss of generality, that all our TBoxes are  in normal form. The subsumption relation induced by $\T$ is denoted by $\sqsubseteq_\T$: we write $S_1\sqsubseteq_\T S_2$ if $\T\models S_1\sqsubseteq S_2$, where $S_1$, $S_2$ are both either concepts or roles. We write $R(a,b)\in \A$ if $P(a,b)\in \A$ and $R= P$, or $P(b,a)\in\A$ and $R = P^-$; we also write $(\exists R)(a)\in \A$ if $R(a,b)\in\A$ for some $b$. An ABox $\A$ is called \emph{H-complete with respect to $\T$} in case 
\begin{align*}
P(a,b) \in \A & \ \ \text{ if } \ R(a,b)\in \A, \text{ for roles } P \text{ and } R\text{ with } R\sqsubseteq_\T P,\\
A(a) \in \A &  \ \ \text{ if } \ B(a)\in \A, \text{ for a concept name } A \text{ and basic concept }  B \text{ with } B\sqsubseteq_\T A.
\end{align*}

A \emph{conjunctive query} (CQ) $\q(\vec{x})$ is a formula $\exists \vec{y}\, \varphi(\vec{x}, \vec{y})$, where $\varphi$ is a conjunction of atoms $A_k(z_1)$ or $P_k(z_1,z_2)$ with $z_i \in \vec{x} \cup \vec{y}$ (without loss of generality, we assume that CQs do not contain constants). 
We denote by $\vars(\q)$ the variables $\vec{x} \cup \vec{y}$ of $\q$ and by $\avars(\q)$ the \emph{answer variables}~$\vec{x}$. 
An \emph{ontology-mediated query} (OMQ) is a pair $\omq(\vec{x}) = (\T,\q(\vec{x}))$, where $\T$ is a TBox and $\q(\vec{x})$ a CQ. A tuple $\vec{a}$ in $\ind (\mathcal{A})$ is a \emph{certain answer} to $\omq(\vec{x})$ over an ABox $\mathcal{A}$ if  $\mathcal{I} \models \q(\vec{a})$ for all models $\mathcal{I}$ of $\T$ and $\A$; in this case we write \mbox{$\T,\A \models \q(\vec{a})$}. If $\vec{x} = \emptyset$, then a certain answer to $\omq$ over $\A$ is `yes' if $\T,\A \models \q$ and `no' otherwise. We often regard a CQ $\q$ as the set of its atoms.

Every consistent \OWLQL{} \emph{knowledge base} (KB) $(\T,\A)$ has a \emph{canonical model}  
$\can$ with the property that 
$\T,\A \models \q(\vec{a})$ iff $\can \models \q(\vec{a})$,
for any CQ $\q$ and any $\vec{a}$ in $\ind(\A)$.
Thus, CQ answering in \OWLQL{} amounts to finding a homomorphism from  the given CQ into the canonical model.  
Informally, $\can$ is obtained from $\A$ by repeatedly applying the axioms in $\T$, 
introducing fresh elements  
as needed to serve as witnesses for the existential quantifiers.  
According to the standard construction  (cf.~\cite{KR10our}), 
the domain $\Delta^{\can}$ of $\can$ consists of words of the form
$a R_1  \dots R_n$ ($n \geq 0$) with $a\in\ind(\A)$ and $R_1 \dots R_n \in \rni^*$ such that 
(\emph{i}) $\T, \A \models \exists R_1(a)$ and (\emph{ii}) $\exists R_i^- \sqsubseteq_\T \exists R_{i+1}$ and $R_i^- \not \sqsubseteq_\T R_{i+1}$, for $1 \leq i < n$.  
We let $\twords$ consist of all words $R_1 \dots R_n \in \rni^*$ satisfying condition~(\emph{ii}). 
A TBox $\T$ is \emph{of depth} $\omega$ if $\twords$ is infinite, and 
 \emph{of depth} $d< \omega$, if $d$ is the maximum length of the words in $\twords$. 

A \emph{datalog program}, $\Pi$, is a finite set of Horn clauses 
$\forall \vec{z}\, (\gamma_0 \leftarrow \gamma_1 \land \dots \land \gamma_m)$,
where each $\gamma_i$ is an atom $S(\vec{y})$ with $\vec{y} \subseteq  \vec{z}$ or an equality $(z = z')$ with $z,z'\in \vec{z}$. (As usual, when writing clauses, we omit $\forall \vec{z}$.) The atom $\gamma_0$ is the \emph{head} of the clause, and $\gamma_1,\dots,\gamma_m$ its  \emph{body}. All variables in the head must also occur in the body, and $=$ can only occur in the body. 
The predicates in the heads of clauses in $\Pi$ are \emph{IDB predicates},  the rest (including~$=$) \emph{EDB predicates}. 
A predicate $S$ \emph{depends} on $S'$ in $\Pi$ if $\Pi$ has a clause with $S$ in the head and $S'$ in the body; $\Pi$ is a \emph{nonrecursive datalog} (NDL) \emph{program} if the (directed)  \emph{dependence graph} of the  dependence relation is acyclic. 

An \emph{NDL query} is a pair $(\Pi,G(\vec{x}))$, where $\Pi$ is an NDL program and $G(\vec{x})$ a predicate. A tuple $\vec{a}$ in $\ind(\A)$ is an \emph{answer to $(\Pi,G(\vec{x}))$ over} an ABox $\A$ if $G(\vec{a})$ holds in the first-order model with domain $\ind(\A)$ obtained by closing $\A$ under the clauses in $\Pi$; in this case we write $\Pi, \A \models G(\vec{a})$. The problem of checking whether  $\vec{a}$ is an answer to $(\Pi,G(\vec{x}))$ over $\A$ is called the \emph{query evaluation problem}. 
The \emph{arity of} $\Pi$ is the maximal arity, $\ari(\Pi)$, of predicates in $\Pi$. The \emph{depth} of $(\Pi,G(\vec{x}))$ is the length, $\dep(\Pi,G)$, of the longest directed path in the dependence graph for $\Pi$ starting from $G$. 
NDL queries are \emph{equivalent} if they have exactly the same answers over any ABox.

An NDL query $(\Pi,G(\vec{x}))$ is an \emph{\NDL-rewriting of an OMQ $\omq(\vec{x}) = (\T,\q(\vec{x}))$ over H-complete ABoxes} in case $\T,\A\models \q(\vec{a})$ iff  $\Pi,\mathcal{A} \models G(\vec{a})$, for any H-complete ABox $\A$ and any tuple $\vec{a}$ in $\ind (\mathcal{A})$.
Rewritings \emph{over arbitrary ABoxes} are defined by dropping the condition that the ABoxes are H-complete. 
Let $(\Pi, G(\vec{x}))$ be an \NDL-rewriting of $\omq(\vec{x})$ over H-complete ABoxes. Denote by $\Pi^*$ the result of replacing each predicate $S$ in $\Pi$ with a fresh predicate $S^*$ and adding the clauses    
$A^*(x) \leftarrow B'(x)$, for $B \sqsubseteq_\T A$, and
$P^*(x,y) \leftarrow R'(x,y)$, for $R \sqsubseteq_\T P$,
where $B'(x)$ and $R'(x,y)$ are the obvious first-order translations of $B$ and $R$ (for example, $B'(x) = \exists y\, R(x,y)$ if $B = \exists R$).
It is easy to see that $(\Pi^*,G^*(\vec{x}))$ is an \NDL-rewriting of $\omq(\vec{x})$  over arbitrary ABoxes. 

It is well-known~\cite{DBLP:journals/ws/CaliGL12} that, without loss of generality, we can only consider NDL-rewritings of  OMQs $(\T,\q(\vec{x}))$ over ABoxes $\A$ that are \emph{consistent} with $\T$.

We call an NDL query $(\Pi,G(x_1,\dots,x_n))$ \emph{ordered} if each of its IDB predicates $S$ comes with fixed variables $x_{i_1},\dots,x_{i_k}$ ($1 \le i_1 < \dots < i_k \le n$), called the \emph{parameters of} $S$,  such that (\emph{i}) every occurrence of $S$ in $\Pi$ is of the form $S(y_1,\dots,y_m,x_{i_1},\dots,x_{i_k})$, (\emph{ii}) the $x_i$ are the parameters of $G$, and (\emph{iii}) if $\vec{x}'$ are all the parameters in the body of a clause, then the head has $\vec{x}'$ among its parameters.  The \emph{width} $\wid(\Pi,G)$ of an ordered\linebreak $(\Pi,G)$ is the maximal number of non-parameter variables in a clause of $\Pi$. All our NDL-rewritings in Secs.~\ref{sec:boundedtw}--\ref{sec:boundedleaf} are ordered, so we now only consider ordered NDL queries.


\section{NL and LOGCFL Fragments of Nonrecursive Datalog}\label{sec3}

In this section, we identify two classes of (ordered) \NDL{} queries with the evaluation problem in the complexity classes \NL{} and \LOGCFL{} for combined complexity.
Recall~\cite{Abitebouletal95} that an NDL program is called \emph{linear} if the body of its every  clause contains at most one IDB predicate (remember that equality is an EDB predicate).
\begin{theorem}\label{linear-nl}
Fix some $\wid > 0$. The combined complexity of evaluating linear NDL queries of width at most $\wid$ is \NL-complete.
\end{theorem}
\begin{proof}
Let $(\Pi,G(\vec{x}))$ be a linear NDL query. 
Deciding whether $\Pi, \A \models G(\vec{a})$ is reducible to finding a path to $G(\vec{a})$ from a certain set $X$ in the \emph{grounding graph} $\mathfrak{G}(\Pi, \A,\vec{a})$ constructed as follows. The vertices of the graph are the ground atoms 
obtained by taking an IDB atom from $\Pi$, replacing each of its parameters by the corresponding constant from $\vec{a}$, and replacing each non-parameter variable by some constant from $\A$. The graph has an edge from $S(\vec{c})$ to $S'(\vec{c}')$
iff the grounding of $\Pi$ contains a clause 
$S'(\vec{c}')\leftarrow S(\vec{c})\land E_1(\vec{e}_1)\land \dots \land E_k(\vec{e}_k)$
with $E_j(\vec{e}_j) \in \A$, for $1\le j \le k$ (we assume that $(c=c) \in \A$).
The set $X$ consists of all vertices $S(\vec{c})$ with IDB predicates $S$ being
of in-degree 0 in the dependency graph of $\Pi$ for which there is a clause 
$S(\vec{c}) \leftarrow E_1(\vec{e}_1)\land \dots \land E_k(\vec{e}_k)$ in the grounding
of $\Pi$ with $E_j(\vec{e}_j) \in \A$ \mbox{($1\le j \le k$)}.
Bounding the width of $(\Pi,G)$ ensures that 
$\mathfrak{G}(\Pi, \A,\vec{a})$ is of polynomial size and can be constructed by a deterministic Turing machine with separate input, write-once output and logarithmic-size working tapes.\qed
\end{proof}

The transformation of NDL-rewritings over H-complete ABoxes into rewritings for arbitrary ABoxes  in Section~\ref{sec:prelims} does not preserve linearity. 
However, we can still show that it suffices to consider the H-complete case: 
\begin{lemma}\label{linear-arbitrary}
For any fixed $\wid > 0$, there is an $\mathsf{L}^\NL$-transducer that, given a linear \NDL-rewriting of an OMQ $\omq(\vec{x})$ over H-complete ABoxes that is of width at most $\wid$, computes a linear NDL-rewriting of $\omq(\vec{x})$ over arbitrary ABoxes whose width is at most $\wid+1$.
\end{lemma}

The complexity class \LOGCFL{} can be defined in terms of \emph{nondeterministic auxiliary pushdown automata}  (NAuxPDAs)~\cite{DBLP:journals/jacm/Cook71}, which are nondeterministic Turing machines with an additional work tape constrained to operate as a pushdown store. Sudborough~\cite{sudborough78} proved that \LOGCFL{} coincides with the class of problems that are solved by NAuxPDAs running in logarithmic space and polynomial time (the space on the pushdown tape is not subject to the logarithmic  bound).

We call an NDL query $(\Pi,G)$ \emph{skinny} if the body of any  clause in $\Pi$ has $\le 2$ atoms. 
\begin{lemma}\label{prop:SkinnyNDLEvaluation}
For any skinny NDL query $(\Pi, G(\vec{x}))$ and ABox $\A$, query evaluation can be done by an NAuxPDA in space $\log |\Pi| + \wid(\Pi,G) \cdot \log |\A|$ and time $2^{O(\dep(\Pi,G))}$.
\end{lemma}
\begin{proof}
Let $\Pi_{\A}^{\vec{a}}$ be the set of ground clauses obtained by first replacing each parameter in $\Pi$ by the corresponding constant from $\vec{a}$, 
and then performing the standard grounding of $\Pi$ using the constants from $\A$. 
Consider the monotone Boolean circuit $\cir(\Pi,\A,\vec{a})$ constructed as follows. 
The output of $\cir(\Pi,\A,\vec{a})$ is $G(\vec{a})$. 
For every atom $\gamma$ occurring in the head of a clause in $\Pi_{\A}^{\vec{a}}$, we take an \textsc{or}-gate whose output is $\gamma$ and inputs are the bodies of the clauses with head~$\gamma$; for every such body, we take an \textsc{and}-gate whose inputs are the atoms in the body. We set an input gate $\gamma$ to 1 iff $\gamma \in \A$. Clearly, $\cir(\Pi,\A,\vec{a})$ is a semi-unbounded fan-in circuit (where \textsc{or}-gates have arbitrarily many inputs, and \textsc{and}-gates two inputs) with $O(|\Pi| \cdot |\A|^{\wid(\Pi,G)})$ gates and depth $O(\dep(\Pi,G))$. 
It is known that the nonuniform analog of \LOGCFL{} can be defined using families of semi-unbounded fan-in circuits of polynomial size and logarithmic depth. Moreover,  there is an algorithm that, given such a circuit $\cir$, computes the output using an NAuxPDA 
in logarithmic space in the size of $\cir$ and exponential time in the depth of $\cir$~\cite[pp.~392--397]{DBLP:journals/jcss/Venkateswaran91}.
Observing that $\cir(\Pi,\A,\vec{a})$ can be computed by a deterministic logspace Turing machine, we conclude that the query evaluation problem can be solved by an NAuxPDA in space $\log |\Pi| + \wid(\Pi,G) \cdot \log |\A|$ and time $2^{O(\dep(\Pi,G))}$.\qed
\end{proof}

A function $\nu$ from the predicate names in $\Pi$ to $\mathbb N$ is a \emph{weight function for} an NDL-query $(\Pi, G(\vec{x}))$ if $\nu(P) > 0$, for any IDB $P$ in $\Pi$, and $\nu(P) \ge \nu(Q_1) + \dots + \nu(Q_n)$, for any   
$P(\vec{z}) \leftarrow Q_1(\vec{z}_1) \land \dots \land Q_n(\vec{z}_n)$ in $\Pi$.  
\begin{lemma}\label{thm:NDLToSkinny}
If $(\Pi, G(\vec{x}))$ has a weight function $\nu$, then it is equivalent to a skinny NDL query $(\Pi',G(\vec{x}))$ such that $|\Pi'|$ is polynomial in $|\Pi|$, $\dep(\Pi',G) \le \dep(\Pi,G) + \log \nu (G)$ and $\wid(\Pi',G) \le \wid(\Pi,G)$. 
\end{lemma}
\begin{proof}
The proof is by induction on $\dep(\Pi,G)$. If $\dep(\Pi,G)=0$, we take $\Pi'=\Pi$. 
Suppose $\Pi$ contains a clause $\psi$ of the form $G(\vec{z}) \leftarrow P_1(\vec{z}_1) \land \dots \land 
P_k(\vec{z}_k)$ and, for each $1\le j \le k$, we have an NDL query $(\Pi'_{P_j},P_j)$ which is equivalent to $(\Pi, P_j)$ and such that
\begin{equation}\label{eq:depPj}
\dep(\Pi'_{P_j},P_j) ~\le~ \dep(\Pi_{P_j},P_j) + \log \nu(P_j) ~\le~ \dep(\Pi,G) - 1 + \log \nu(P_j).
\end{equation}
We construct the Huffman tree~\cite{huf52} for the alphabet $\{1,\dots,k\}$, where the frequency of $j$ is $\nu(P_j)/\nu(G)$ (by definition, $\nu(G) >0$). The Huffman tree is  binary and has $k$ leaves, denoted $1,\dots,k$,  and $k-1$ internal nodes (including the root, $g$), and the length of the path from $g$ to any leaf $j$ at most $\lceil \log(\nu(G)/\nu(P_j))\rceil$. For each internal node $v$ of the tree (but the root), we take a predicate $P_v(\vec{z}_v)$, where $\vec{z}_v$ is the union of $\vec{z}_u$ for all descendants $u$ of $v$; for the root $g$, we take $P_g(\vec{z}_g) = G(\vec{z})$. Let $\Pi'_\psi$ be the extension of the union of $\Pi'_{P_j}$, for $1\leq j \leq k$, with clauses $P_v(\vec{z}_{v}) \leftarrow P_{u_1}(\vec{z}_{u_1}) \land P_{u_2}(\vec{z}_{u_2})$, for each $v$ with immediate successors  $u_1$ and $u_2$.
 The number of the new clauses is $k-1$. Consider the NDL query $(\Pi'_\psi,G(\vec{z}))$. By~\eqref{eq:depPj}, we  have:
\begin{multline*}
\dep(\Pi'_\psi,G) \le \max\nolimits_j \{ \lceil \log (\nu(G) / \nu(P_j)) \rceil + \dep(\Pi'_{P_j},P_j) \} \le{} \\
\max\nolimits_j \{ \log (\nu(G) / \nu(P_j)) + \dep(\Pi,G) + \log \nu(P_j) \} =
\log \nu(G) + \dep(\Pi,G).
\end{multline*}
Let $\Pi'$ be the result of applying this transformation to
each clause in $\Pi$ with head $G(\vec{z})$. 
It is readily seen that $(\Pi',G)$ is as required; in particular, $|\Pi'|= O(|\Pi|^2)$. \qed 
\end{proof}

\begin{theorem} Fix $c \ge 1$, $\wid \ge 1$ and a polynomial $p$. Query evaluation for NDL queries $(\Pi, G(\vec{x}))$ with a weight function $\nu$ such that $\nu(G) \le p(|\Pi|)$, $\wid(\Pi,G) \le \wid$ and  
$\dep(\Pi,G) \le c \log \nu(G)$ is in \LOGCFL{} for combined complexity.
\end{theorem}
\begin{proof}
By Lemma~\ref{thm:NDLToSkinny}, $(\Pi,G)$ is equivalent to a skinny NDL query $(\Pi',G')$ with $|\Pi'|$ polynomial in $|\Pi|$, $\wid(\Pi',G) \le \wid$, and $\dep(\Pi',G') \le (c + 1)\log \nu(G)$. By  
Lemma~\ref{prop:SkinnyNDLEvaluation}, query evaluation for $(\Pi',G')$ over $\A$ is   solved by an NAuxPDA in space $\log |\Pi'| + \wid(\Pi',G) \cdot \log |\A| = O(\log |\Pi|+ \log |\A|)$ and time $2^{O(\dep(\Pi',G'))} \le 2^{O(\log \nu(G))} = (\nu(G))^{O(1)} \le p'(|\Pi|)$, for some polynomial $p'$.\qed
\end{proof}

\begin{corollary}\label{cor:weight}
Suppose there is an algorithm that, given any OMQ $\omq(\vec{x})$ from some class~$\mathcal{C}$, constructs its NDL-rewriting $(\Pi,G(\vec{x}))$ over H-complete ABoxes having a weight function $\nu$ with $\nu(G) \le |\omq|$ and  $\dep(\Pi,G) \le c \log \nu(G)$, and such that $\wid(\Pi,G) \le \wid$ and $|\omq| \le |\Pi| \le p(|\omq|)$, for some fixed constants $c$, $\wid$ and polynomial $p$. 
Then the evaluation problem for the NDL-rewritings $(\Pi^*,G^*(\vec{x}))$ of the OMQs in $\mathcal{C}$ over arbitrary ABoxes \textup{(}defined in Section~\ref{sec:prelims}\textup{)} is in \LOGCFL{} for combined complexity.
\end{corollary}


\section{Bounded Treewidth CQs and Bounded-Depth TBoxes}\label{sec:boundedtw}

With every CQ $\q$, we associate its \emph{Gaifman graph} $\gfmn$ whose vertices are the variables of $\q$ and edges are the  pairs $\{u,v\}$ such that $P(u,v)\in\q$, for some $P$. We call $\q$ \emph{tree-shaped} if $\gfmn$ is a tree; $\q$ is \emph{connected} if the graph $\gfmn$ is connected.
A \emph{tree decomposition} of an undirected graph $\gfmn=(V,E)$ is a pair $(T,\lambda)$, where $T$ is an (undirected) tree and $\lambda$ a function from the set of nodes of $T$ to $2^V$ such that 
the following conditions hold:
\begin{itemize}
\item[--] for every $v \in V$, there exists a node $\nd$ with $v \in \lambda(\nd)$;

\item[--] for every $e \in E$, there exists a node $\nd$ with $e \subseteq \lambda(\nd)$;

\item[--] for every $v \in V$, the nodes $\{\nd\mid v \in \lambda(\nd)\} $ induce a connected subtree of~$T$.
\end{itemize}
We call the set $\lambda(\nd) \subseteq V$ a \emph{bag for} $\nd$. The \emph{width} of $(T, \lambda)$ is $\max_{\nd\in T} |\lambda(\nd)| - 1$. The \emph{treewidth of a graph} $\gfmn$ is the minimum width over all tree decompositions of $\gfmn$. The \emph{treewidth of a CQ} is the treewidth of its Gaifman graph.

\begin{example}\label{ex:rewriting:1}
Consider CQ $\q(x_0, x_7)$ depicted below (black nodes are answer variables):\\
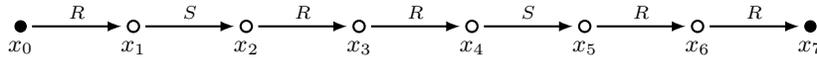
\centerline{\begin{tikzpicture}[>=latex]\footnotesize
\node[bpoint,label=below:{$x_0$}] (v0) at (0,0) {};
\node[wpoint,label=below:{$x_1$}] (v1) at (1.5,0) {};
\node[wpoint,label=below:{$x_2$}] (v2) at (3,0) {};
\node[wpoint,label=below:{$x_3$}] (v3) at (4.5,0) {};
\node[wpoint,label=below:{$x_4$}] (v4) at (6,0) {};
\node[wpoint,label=below:{$x_5$}] (v5) at (7.5,0) {};
\node[wpoint,label=below:{$x_6$}] (v6) at (9,0) {};
\node[bpoint,label=below:{$x_7$}] (v7) at (10.5,0) {};
\begin{scope}[thick,shorten >= 2pt, shorten <= 2pt]\scriptsize
\draw[->] (v0) to node[above] {$R$} (v1);
\draw[->] (v1)to node[above] {$S$}  (v2);
\draw[->] (v2) to node[above] {$R$}  (v3);
\draw[->] (v3) to node[above] {$R$}  (v4);
\draw[->] (v4) to node[above] {$S$}  (v5);
\draw[->] (v5) to node[above] {$R$}  (v6);
\draw[->] (v6) to node[above] {$R$}  (v7);
\end{scope}
\end{tikzpicture}}\\
Its natural tree decomposition of treewidth 1 is based on the the chain  $T$ of 7~vertices, which are represented as bags as follows:\\[2pt]
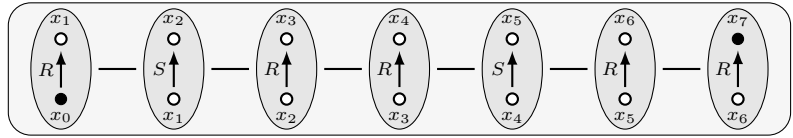
\centerline{\begin{tikzpicture}[>=latex,yscale=0.8]\scriptsize
\draw[rounded corners=3mm,fill=gray!7] (-0.7,-0.6) rectangle (9.7,1.6);
\foreach \x in {0,1.5,3,4.5,6,7.5,9} {
\draw[fill=gray!20,thin] (\x,0.5) ellipse (0.4 and 1);
}
\foreach \x in {0,1.5,3,4.5,6,7.5} {
\draw[thick] (\x+0.5,0.5) -- +(0.5,0);
}
\node[bpoint,label=below:{$x_0$}] (v0) at (0,0) {};
\node[wpoint,label=above:{$x_1$}] (v1p) at (0,1) {};
\node[wpoint,label=below:{$x_1$}] (v1) at (1.5,0) {};
\node[wpoint,label=above:{$x_2$}] (v2p) at (1.5,1) {};
\node[wpoint,label=below:{$x_2$}] (v2) at (3,0) {};
\node[wpoint,label=above:{$x_3$}] (v3p) at (3,1) {};
\node[wpoint,label=below:{$x_3$}] (v3) at (4.5,0) {};
\node[wpoint,label=above:{$x_4$}] (v4p) at (4.5,1) {};
\node[wpoint,label=below:{$x_4$}] (v4) at (6,0) {};
\node[wpoint,label=above:{$x_5$}] (v5p) at (6,1) {};
\node[wpoint,label=below:{$x_5$}] (v5) at (7.5,0) {};
\node[wpoint,label=above:{$x_6$}] (v6p) at (7.5,1) {};
\node[wpoint,label=below:{$x_6$}] (v6) at (9,0) {};
\node[bpoint,label=above:{$x_7$}] (v7p) at (9,1) {};
\begin{scope}[thick,shorten >= 2pt, shorten <= 2pt]\scriptsize
\draw[->] (v0) to node[left] {$R$} (v1p);
\draw[->] (v1)to node[left] {$S$}  (v2p);
\draw[->] (v2) to node[left] {$R$}  (v3p);
\draw[->] (v3) to node[left] {$R$}  (v4p);
\draw[->] (v4) to node[left] {$S$}  (v5p);
\draw[->] (v5) to node[left] {$R$}  (v6p);
\draw[->] (v6) to node[left] {$R$}  (v7p);
\end{scope}
\end{tikzpicture}}
\end{example}

Fix a connected CQ $\q(\vec{x})$ and a tree decomposition $(T, \lambda)$ of its Gaifman graph $\gfmn = (V,E)$. 
Let $D$ be a subtree of $T$. The \emph{size} of $D$ is the number of nodes in it.
We call a node $\nd$ of $D$ \emph{boundary} if $T$ has an edge $\{\nd,\nd'\}$ with $\nd'\notin D$, and let the \emph{degree} $\degree(D)$ of $D$ be the number of its boundary nodes. Note that $T$ itself  is the only subtree of $T$ of degree  $0$.  We say that a node $\nd$ \emph{splits} $D$ into subtrees $D_1,\dots,D_k$ if the $D_i$ partition $D$ without $\nd$: each node of $D$ except $\nd$ belongs to exactly one $D_i$.
\begin{lemma}[\cite{DBLP:conf/lics/BienvenuKP15}]\label{l:6.8}
Let $D$ be a subtree of $T$ of size $m > 1$.\\ 
If $\degree(D) =2$, then there is a node $\nd$ splitting $D$ into subtrees of size $\leq m/2 $ and degree~$\leq 2$ and, possibly, one subtree of size $<m-1$ and degree~$1$.\\ 
If $\degree(D) \leq 1$, then there is $\nd$ splitting $D$ into subtrees of size $\leq m/2 $ and degree $\leq 2$.
\end{lemma}

In Example~\ref{ex:rewriting:1}, $t$ splits $T$ into $T_1$ and $T_2$ depicted below:\\[2pt]
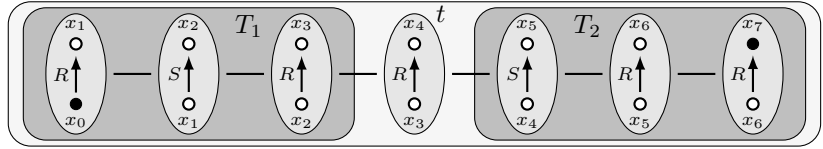
\centerline{\begin{tikzpicture}[>=latex,yscale=0.8]\scriptsize
\draw[rounded corners=3mm,fill=gray!7] (-0.9,-0.7) rectangle (9.9,1.7);
\draw[rounded corners=3mm,fill=gray!50] (-0.7,-0.6) rectangle (3.7,1.6);
\node at (2.3,1.3) {\normalsize $T_1$};
\draw[rounded corners=3mm,fill=gray!50] (5.3,-0.6) rectangle (9.7,1.6);
\node at (6.8,1.3) {\normalsize $T_2$};
\node at (4.85,1.5) {\normalsize $t$};
\foreach \x in {0,1.5,3,4.5,6,7.5,9} {
\draw[fill=gray!20,thin] (\x,0.5) ellipse (0.4 and 1);
}
\foreach \x in {0,1.5,3,4.5,6,7.5} {
\draw[thick] (\x+0.5,0.5) -- +(0.5,0);
}
\node[bpoint,label=below:{$x_0$}] (v0) at (0,0) {};
\node[wpoint,label=above:{$x_1$}] (v1p) at (0,1) {};
\node[wpoint,label=below:{$x_1$}] (v1) at (1.5,0) {};
\node[wpoint,label=above:{$x_2$}] (v2p) at (1.5,1) {};
\node[wpoint,label=below:{$x_2$}] (v2) at (3,0) {};
\node[wpoint,label=above:{$x_3$}] (v3p) at (3,1) {};
\node[wpoint,label=below:{$x_3$}] (v3) at (4.5,0) {};
\node[wpoint,label=above:{$x_4$}] (v4p) at (4.5,1) {};
\node[wpoint,label=below:{$x_4$}] (v4) at (6,0) {};
\node[wpoint,label=above:{$x_5$}] (v5p) at (6,1) {};
\node[wpoint,label=below:{$x_5$}] (v5) at (7.5,0) {};
\node[wpoint,label=above:{$x_6$}] (v6p) at (7.5,1) {};
\node[wpoint,label=below:{$x_6$}] (v6) at (9,0) {};
\node[bpoint,label=above:{$x_7$}] (v7p) at (9,1) {};
\begin{scope}[thick,shorten >= 2pt, shorten <= 2pt]\scriptsize
\draw[->] (v0) to node[left] {$R$} (v1p);
\draw[->] (v1)to node[left] {$S$}  (v2p);
\draw[->] (v2) to node[left] {$R$}  (v3p);
\draw[->] (v3) to node[left] {$R$}  (v4p);
\draw[->] (v4) to node[left] {$S$}  (v5p);
\draw[->] (v5) to node[left] {$R$}  (v6p);
\draw[->] (v6) to node[left] {$R$}  (v7p);
\end{scope}
\end{tikzpicture}}

\smallskip

We define recursively a set $\R$ of subtrees of $T$, a binary relation $\prec$ on $\R$  and a function $\sigma$ on $\R$ indicating the splitting node. We begin by adding $T$ to $\R$. Take $D\in \R$ that has not been split yet. If $D$  is of size~1 then let $\sigma(D)$ be the only node of $D$. Otherwise, by Lemma~\ref{l:6.8}, we find a node $\nd$ in $D$ that splits it into $D_1,\dots,D_k$. We set $\sigma(D) = t$ and, for each $1\leq i\leq k$, add  $D_i$ to $\R$ and set $D_i \prec D$; then, we apply the procedure recursively to each of $D_1,\dots,D_k$.
In Example~\ref{ex:rewriting:1} with $\nd$ splitting $T$, we have $\sigma(T) = t$, $T_1 \prec T$ and 
$T_2 \prec T$.

For each $D\in\R$, we recursively define a set of atoms $\q_D$ by taking 
$$
\q_D \ \ = \ \ \bigl\{S(\vec{v}) \in \q \mid \vec{v} \subseteq \lambda(\sigma(D)) \bigr\} \ \cup \ \bigcup\nolimits_{D' \prec D} \q_{D'}.
$$ 
By the definition of tree decomposition, $\q_T = \q$. Denote by $\vec{x}_D$ the subset of  $\avars(\q)$ that occur in $\q_D$. In our running example, $\vec{x}_{T} = \{x_0, x_7\}$,  $\vec{x}_{T_1} = \{x_0\}$ 
and $\vec{x}_{T_2} = \{x_7\}$.
Denote by $\dD$  the union of all $\lambda(\nd) \cap\lambda(\nd')$ for
a boundary node $\nd$ of $D$ and its unique neighbour $\nd'$ in $T$ \emph{outside} $D$.  If $D$ is a singleton $\{d\}$, then $\partial D$ consists of those variables in $\lambda(d)$ that occur in at least one other bag. 
In our example, $\partial T=\emptyset$, $\partial T_1 =\{x_3\}$ and $\partial T_2 =\{x_4\}$.

Let $\T$ be a TBox of finite depth $k$. A \emph{type} is a partial map $\vec{w}$ from 
$V$ to $\twords$; its domain is denoted by $\dom(\vec{w})$. By  $\vec{\varepsilon}$ we denote the unique partial type with \mbox{$\dom(\vec{\varepsilon}) = \emptyset$}.
We use types to represent how variables are mapped into $\can$, with $\vec{w}(u)=w$ indicating that $u$ is mapped to an element of the form $a w$ (for some $a \in \ind(\A)$), and 
with $\vec{w}(u)=\varepsilon$ that $u$ is mapped to an ABox individual. 
We say that a type $\vec{w}$ is \emph{compatible} with
a bag $\nd$ if, for all $u,v\in\lambda(\nd)\cap\dom(\vec{w})$, we have
\begin{itemize}\itemsep=0pt
\item  if $v \in \avars(\q)$, then $\vec{w}(v) = \varepsilon$; 
\item if $A(v) \in \q$, then either $\vec{w}(v)=\varepsilon$ or
$\vec{w}(v)= w R \text{ with } \exists R^- \sqsubseteq_\T A$;
\item if $R(v, u)\in \q$, then
$\vec{w}(v) = \vec{w}(u) =\varepsilon$, or \mbox{$\vec{w}(u) = \vec{w}(v) R'$} with $R' \sqsubseteq_\T R$, 
or $\vec{w}(v) = \vec{w}(u) R'$ with  $R' \sqsubseteq_\T R^-$.
\end{itemize}

In the sequel, we abuse notation and use sets of variables in place of sequences assuming that they are ordered in some (fixed) way. For example, we use $\vec{x}_D$ for a tuple of variables in the set $\vec{x}_D$ (ordered in some way). Also, given a tuple $\vec{a}$ in $\ind(\A)$ of length $|\vec{x}_D|$   and $x\in\vec{x}_D$, we write $\vec{a}(x)$ to refer to the element of $\vec{a}$ that corresponds to $x$ (that is, to the component of the tuple with the same index).

Let $\Pi_\omq$ be an NDL program that---for any $D\in\R$, any types $\tpd$ and $\tpr$ for which $\dom (\tpd)=\dD$, $\dom(\tpr) = \lambda(\sigma(D))$, $\tpr$ is compatible with $\sigma(D)$ and agrees with $\tpd$ on their common domain---contains the clause
\begin{equation}
\rpred^{\tpd}_D(\dD, \vec{x}_D) \ \ \leftarrow \ \ \mathsf{At}^{\tpr} \ \land 
\bigwedge\nolimits_{D' \prec D} \rpred^{(\tpr\cup\tpd) \restr\dDp}_{D'}(\dDp,\vec{x}_{D'}),
\end{equation}
where $\vec{x}_D$ are the parameters of predicate $\rpred^{\tpd}_D$, 
$(\tpr\cup\tpd) \restr\dDp$ is
the restriction\footnote{
By construction,  
$\dom(\tpr\cup\tpd)$  covers $\dDp$, and so  
the domain of
$(\tpr\cup\tpd) \restr \dDp$ is $\dDp$. 
} of the union $\tpr\cup\tpd$ of $\tpr$ and $\tpd$ to $\dDp$,
and $\mathsf{At}^{\tpr}$ is defined as follows:
\begin{equation}\label{eq:at}
\mathsf{At}^{\tpr} \ \ \ =\ \ \bigwedge_{\substack{A(u)\in\q\\\tpr(u) = \varepsilon}} \hspace*{-0.5em}A(u) \ \ \ \land 
\bigwedge_{\substack{R(u,v)\in \q\\ \tpr(u) = \tpr(v) = \varepsilon}} \hspace*{-1.5em}R(u, v) \ \ \ \land
\bigwedge_{\substack{R(u,v)\in \q\\ \tpr(u) \ne\varepsilon \text{ or } \tpr(v) \ne \varepsilon}} \hspace*{-2em} (u =  v) \ \ \ \land
\bigwedge_{\substack{\tpr(u) = Sw'\\\text{ for some } w'}} \hspace*{-1em}A_S(u).
\end{equation}
The first two conjunctions in $\mathsf{At}^{\tpr}$ ensure that atoms all of whose variables are assigned $\varepsilon$
are present in the ABox. The third conjunction ensures that if one of the variables in a role atom is not mapped to $\varepsilon$,
then the images of the variables share the same initial individual. 
Finally, atoms in the final conjunction ensure that if a variable is to be mapped to $aSw'$,
then the individual $a$ satisfies $\exists S$ (so $aSw'$ is part of the domain of $\can$).

\begin{example}\label{ex:rewriting:2}
Now we fix an ontology $\T$ with the following axioms:
\begin{align*}
 A & \equiv  \exists P, &  P & \sqsubseteq S, & P & \sqsubseteq R^-, &\qquad
B & \equiv \exists Q, & Q & \sqsubseteq R, & Q & \sqsubseteq S^-. 
\end{align*}
Since $\lambda(\nd) = \{ x_3, x_4\}$, there are only three types compatible with $\nd$:
\begin{equation*}
\tpr_1\colon x_3\mapsto \varepsilon, x_4\mapsto \varepsilon,\qquad \tpr_2\colon x_3 \mapsto P, x_4\mapsto \varepsilon\quad \text{ and }\quad 
\tpr_3\colon x_3 \mapsto \varepsilon, x_4\mapsto Q. 
\end{equation*}
So,
$\mathsf{At}^{\tpr_1}  = R(x_3, x_4)$,
$\mathsf{At}^{\tpr_2}  = A(x_3) \land (x_3 = x_4)$, 
$\mathsf{At}^{\tpr_3}  = B(x_4) \land (x_3 = x_4)$.
Thus, predicate $\rpred^{\vec{\varepsilon}}_{T}$ is defined by the following clauses, for $\tpr_1$, $\tpr_2$ and $\tpr_3$, respectively:
\begin{align*}
\rpred^{\vec{\varepsilon}}_{T} (x_0,x_7) &\leftarrow \rpred^{x_3\mapsto\varepsilon}_{T_1} (x_3, x_0)\land R(x_3, x_4)\land \rpred^{x_4\mapsto\varepsilon}_{T_2}(x_4, x_7),\\
\rpred^{\vec{\varepsilon}}_{T} (x_0,x_7) &\leftarrow \rpred^{x_3\mapsto P}_{T_1}  (x_3,x_0)\land A(x_3)\land (x_3 =  x_4)\land \rpred^{x_4\mapsto\varepsilon}_{T_2}(x_4, x_7),\\
\rpred^{\vec{\varepsilon}}_{T} (x_0,x_7) &\leftarrow \rpred^{x_3\mapsto\varepsilon}_{T_1} (x_3,x_0)\land B(x_4)\land (x_3 =  x_4)\land\rpred^{x_4\mapsto Q}_{T_2} (x_4, x_7).
\end{align*}
\end{example}

By induction on $\prec$ on $\R$, 
we show that $(\Pi_\omq, \rpred^{\vec{\varepsilon}}_T)$  is a rewriting of $\omq(\vec{x})$.
\begin{lemma}\label{prop:logdepth:rewriting}
For any ABox $\A$, any $D \in \R$, any type $\tpd$ with $\dom(\tpd) = \dD$, any $\vec{b} \in \ind(\A)^{|\dD|}$ and $\vec{a}\in \ind(\A)^{|\vec{x}_D|}$, we have $\Pi_\omq,\A \models \rpred^{\tpd}_D(\vec{b}, \vec{a})$  iff there is a homomorphism $h\colon \q_D \to \can$ such that  
\begin{equation*}
h(x) = \vec{a}(x), \text{ for } x\in \vec{x}_D, \quad\text{ and }\quad 
h(v) = \vec{b}(v) \tpd(v), \text{ for  } v\in \dD.
\end{equation*}
\end{lemma}

Fix now $k$ and $t$, and consider the class of OMQs $\omq(\vec{x}) = (\T,\q(\vec{x}))$ with $\T$ of depth~$\leq k$ and $\q$ of treewidth  $\leq t$. Let $T$ be a tree decomposition of $\q$ of treewidth~$\leq t$. We take the following weight function:  $\nu(\rpred^{\tpd}_D) = |D|$. Clearly, $\nu(\rpred^{\vec{\varepsilon}}_T) \le |\omq|$.  By Lemma~\ref{l:6.8}, 
$\dep(\Pi_\omq,\rpred^{\vec{\varepsilon}}_T) \le 2 \log |T| = 2 \log \nu(\rpred^{\vec{\varepsilon}}_T) \le 2 \log |\omq|$.
Since $|\R| \le |T|^2$ and
there are at most $|\T|^{2tk}$ options for $\tpd$, there are polynomially
many predicates $\rpred^{\tpd}_D$, and so $\Pi_\omq$ is of polynomial size. Thus, by Corollary~\ref{cor:weight}, the obtained NDL-rewriting over arbitrary ABoxes can be evaluated in \LOGCFL{}. Finally, we note that a tree decomposition of treewidth $\leq t$ can be computed using an $\textsf{L}^{\LOGCFL}$-transducer~\cite{DBLP:conf/icalp/GottlobLS99}, and so the NDL-rewriting can also be constructed by an $\textsf{L}^{\LOGCFL}$-transducer.

\section{Bounded-Leaf CQs and Bounded-Depth TBoxes}\label{sec:5}

We next consider OMQs with tree-shaped CQs
in which both the depth of the ontology and the number of leaves in the CQ are bounded. 
Let $\T$ be a TBox of finite depth $k$, 
and let $\q(\vec{x})$ be a tree-shaped CQ with at most $\ell$ leaves.
Fix one of the variables of $\q$ as root, and let $M$ be the maximal distance to a leaf from the root. 
For $n \leq M$, let 
$\vec{z}^n$ denote the set of all variables of $\q$ at distance $n$
from the root; clearly, $|\vec{z}^n| \le \ell$. 
We call the $\vec{z}^n$ \emph{slices} of $\q$ and observe that 
they satisfy the following: for every 
$R(u,v) \in \q$ with $u \neq v$, there exists $0 \leq n < M$ such that either $u\in \vec{z}^n$ and $v\in \vec{z}^{n+1}$ or $u\in\vec{z}^{n+1}$ and $v\in \vec{z}^n$.
For $0 \leq n \leq M$, we denote by $\q_n(\vec{z}^n_\exs, \vec{x}^n)$ the query
consisting of all atoms $S(\vec{u})$ of $\q$ such that $\vec{u} \subseteq \bigcup_{n \leq m \leq M} \vec{z}^m$, 
where 
$\vec{x}^n=\vars(\q_n) \cap \vec{x}$ and
$\vec{z}^n_\exs = \vec{z}^n \setminus \vec{x} $. 

By \emph{type of a slice} $\vec{z}^n$, we mean 
a total map $\vec{w}$ from $\vec{z}^n$ to $\twords$. 
Analogously to Section~\ref{sec:boundedtw}, 
we define what it means for a type (or pair of types) to be compatible 
with a slice (pair of adjacent slices). 
We call $\vec{w}$ \emph{locally compatible} with $\vec{z}^{n}$ if for every $z \in \vec{z}^n$:
\begin{itemize}
\item  if $z \in \avars(\q)$, then $\vec{w}(z) = \varepsilon$; 
\item if $A(z) \in \q$, then either $\vec{w}(z)= \varepsilon$ or
$\vec{w}(z)= w R$ with $\exists R^-  \sqsubseteq_\T A$;
\item if $R(z, z) \in \q$, then $\vec{w}(z) = \varepsilon$. 
\end{itemize}
If $\tpd, \tpr$ are types for $\vec{z}^{n}$ and $\vec{z}^{n+1}$ respectively, 
then we call $(\tpd, \tpr)$ \emph{compatible} with 
$(\vec{z}^{n}, \vec{z}^{n+1})$ if 
 $\vec{w}$ is locally compatible with $\vec{z}^{n}$, 
$\tpr$ is locally compatible with $\vec{z}^{n+1}$, and
for every atom 
$R(z^n, z^{n+1}) \in \q$, one of the following holds: 
(\emph{i}) $
\tpd(z^{n}) =\tpr(z^{n+1})
= \varepsilon$, 
(\emph{ii}) $\tpr(z^{n+1})= \tpd(z^{n}) R'$ with $ R' \sqsubseteq_\T R$, or
(\emph{iii}) $\tpd(z^{n}) = \tpr(z^{n+1}) R' $ with $R' \sqsubseteq_\T R^-$.

Consider the \NDL\ program $\Pi_\omq'$ defined as follows. 
For every $0 \leq n < M$
and every pair of types $(\vec{w}, \tpr)$ that is compatible with $(\vec{z}^n, \vec{z}^{n+1})$, we include the clause: 
\begin{equation*}
P^{\vec{w}}_{n}(\vec{z}^{n}_\exs, \vec{x}^n) \leftarrow 
\mathsf{At}^{\vec{w}\cup\tpr}(\vec{z}^n,\vec{z}^{n+1})  \land P^{\tpr}_{n+1}(\vec{z}^{n+1}_\exs, \vec{x}^{n+1}),
\end{equation*}
where $\vec{x}^n$ are the parameters of 
$P^{\vec{w}}_{n}$ and $\mathsf{At}^{\vec{w}\cup\tpr}(\vec{z}^n,\vec{z}^{n+1})$ 
is the conjunction of atoms~\eqref{eq:at}, as defined in Section~\ref{sec:boundedtw}, for the union $\vec{w}\cup\tpr$ of types $\vec{w}$ and $\tpr$.

For every type $\vec{w}$ locally compatible with $\vec{z}^M$, we include the clause: 
\begin{equation*}
P^{\vec{w}}_{M}(\vec{z}^{M}_\exs, \vec{x}^M) \leftarrow \mathsf{At}^{\vec{w}}(\vec{z}^M). 
\end{equation*}
(Recall that $\vec{z}^M$ is a disjoint union of $\vec{z}^M_\exs$ and $\vec{x}^M$.) We use $G$, with parameters $\vec{x}$, as the goal predicate and include
 $G(\vec{x}) \leftarrow P^{\vec{w}}_{0}(\vec{z}^{0}, \vec{x})$ 
for every predicate $P^{\vec{w}}_{0}(\vec{z}^{0}, \vec{x}^0) $ occurring in the head of one of the preceding clauses. 

The following lemma (which is proved by induction) is the key step in showing that $(\Pi_\omq',G(\vec{x}))$
is a rewriting of $(\T, \q)$ over H-complete ABoxes: 
\begin{lemma}\label{nl-rewriting-correct}
For any H-complete ABox $\A$, any $0 \leq n \leq M$, any predicate $P^{\vec{w}}_n$, 
any $\vec{b} \in \ind(\A)^{|\vec{z}^n_\exists|}$ and any $\vec{a}\in \ind(\A)^{|\vec{x}^n|}$,
we have  
$\Pi_\omq',\A \models P^{\vec{w}}_n(\vec{b}, \vec{a})$  iff there is a homomorphism 
$h\colon \q_n \to \can$ such that  
\begin{equation}\label{nl-rewriting-eq}
h(x) = \vec{a}(x), \text{ for } x\in \vec{x}^n, \quad\text{ and }\quad 
h(z) = \vec{b}(z) \vec{w}(z), \text{ for  } z\in\vec{z}^n_\exs.
\end{equation}
\end{lemma}

It should be clear that $\Pi_\omq'$ is a linear NDL program 
of  width at most~$2 \ell$. Moreover, when $\ell$ and $k$ are bounded by fixed constants, 
it takes only logarithmic space to store a type $\vec{w}$, which allows us to show
that 
 $\Pi'_\omq$ can be computed 
by an $\mathsf{L}^\NL$-transducer. We can apply Lemma~\ref{linear-arbitrary} to obtain 
an NDL rewriting for arbitrary ABoxes, 
and then use Theorem~\ref{linear-nl} to conclude that the resulting 
program can be evaluated in \NL.

\section{Bounded-Leaf CQs and Arbitrary TBoxes}\label{sec:boundedleaf}

For OMQs with bounded-leaf CQs and ontologies of unbounded depth,
our rewriting utilises the notion of tree witness~\cite{DBLP:conf/kr/KikotKZ12}.  
Let $\omq(\vec{x}) = (\T,\q(\vec{x}))$ with $\q(\vec{x}) = \exists \vec{y}\, \varphi(\vec{x}, \vec{y})$. 
For a pair $\t = (\tr, \ti)$ of disjoint sets of variables in $\q$, with $\ti\subseteq \vec{y}$
and $\ti \ne\emptyset$, set 
\begin{equation*}
\q_\t \ = \ \bigl\{\, S(\vec{z}) \in \q \mid \vec{z} \subseteq \tr\cup \ti \text{ and } \vec{z}\not\subseteq \tr\,\bigr\}.
\end{equation*}
If $\q_\t$ is a minimal subset of $\q$ for which there is a homomorphism $h \colon \q_\t  \to \C_\T^{\smash{A_R(a)}}$ such that $\tr = h^{-1}(a)$ and $\q_\t$ contains every atom of $\q$ with at least one variable from $\ti$, then we call $\t = (\tr, \ti)$ a \emph{tree witness for $\omq$ generated by $R$}. 
Note that the same tree witness $\t$ can be generated by different roles $R$.

The logarithmic-depth NDL-rewriting for  bounded-leaf queries and ontologies
of unbounded depth is based upon the following observation~\cite{LICS14}. 
\begin{lemma}\label{PrepMiddleVertex}
Every tree $T$ of size $m$ has a node splitting it into subtrees of size  
$\leq\! \lceil m/2 \rceil$.
\end{lemma}

We will use repeated applications of this lemma to decompose the input CQ into smaller and smaller subqueries. 
Formally, for every tree-shaped CQ $\q$, we use $v_\q$ to denote a vertex in the Gaifman graph $\gfmn$ of $\q$  
that satisfies the condition of Lemma~\ref{PrepMiddleVertex}.
If $|\vars(\q)| = 2$ and $\q$ has at least one existential variable, 
we assume that $v_\q$ is existentially quantified.
Then, for an OMQ $\omq=(\T, \q_0(\vec{x}))$, we define $\sqset$ as the smallest set of queries 
that contains $\q_0(\vec{x})$ and is such that, for every $\q(\vec{z}) \in \sqset$ with $\vars(\q) \ne \vec{z}$,
the following queries also belong to $\sqset$:\label{page:decomposition}
\begin{itemize}
\item for every $u_i$ adjacent to $v_\q$ in $\gfmn$, the query $\q_i(\vec{z}_i)$ 
comprising all role atoms linking $v_\q$ and $u_i$,
as well as all atoms whose variables cannot reach $v_\q$ in $\gfmn$ without passing by $u_i$,
and with $\vec{z}_i = \vars(\q_i) \cap (\vec{z} \cup \{v_\q\})$;
\item for every tree witness $\t$ for $(\T,\q(\vec{z}))$ with  $\tr\neq \emptyset$ and $v_\q\in\ti$, the queries\linebreak $\q_1^\t(\vec{z}_1^\t), \dots, \q_m^\t(\vec{z}_m^\t)$ that correspond 
to the connected components of the set of atoms of $\q$ that are not in  $\q_\t$, with $\vec{z}_i^\t = \vars(\q_i^\t) \cap (\vec{z} \cup \tr)$. 
\end{itemize}

The NDL program $\Pi''_{\omq}$ uses IDB predicates $P_\q$, for $\q(\vec{z}) \in \sqset$,
with arity $|\vec{z}|$ and parameters $\vars(\q) \cap \vec{x}$.  For each $\q(\vec{z}) \in \sqset$ with $\vars(\q) = \vec{z}$, we include the clause $\rew_\q(\vec{z}) \leftarrow \q(\vec{z})$. 
For each $\q(\vec{z}) \in \sqset$ with $\vars(\q) \ne \vec{z}$, we include the clause 
\begin{equation*}
\rew_\q(\vec{z}) \ \ \ \leftarrow \bigwedge_{A(v_\q)\in \q}\hspace*{-1em} A(v_\q) \ \ \ \land \bigwedge_{R(v_\q,v_\q) \in \q} \hspace*{-1em}R(v_\q,v_\q)  \ \ \  \land \bigwedge_{1\leq i\leq n} \rew_{\q_i}(\vec{z}_i), 
\end{equation*}
where $\q_1(\vec{z}_1), \ldots, \q_n(\vec{z}_n)$ are the subqueries induced by the neighbours of $v_\q$ in $\gfmn$, and the following clause
\begin{equation*}
\rew_\q(\vec{z}) \ \ \ \leftarrow \bigwedge_{u,u' \in \tr} (u=u') \ \ \ \land \ \bigwedge_{u \in \tr} A_R(u)\  \
\land \bigwedge_{1\leq i \leq m} \rew_{\q_i^\t}(\vec{z}_i^\t)
\end{equation*}
for every tree witness $\t$ for $(\T,\q(\vec{z}))$ with $\tr\neq \emptyset$ and $v_\q \in\ti$ and for every role~$R$ generating $\t$, 
where $\q_1^\t, \dots, \q_m^\t$ are the connected components of $\q$ without $\q_\t$.
Finally, if $\q_0$ is Boolean, then we additionally include clauses 
$P_{\q_0} \leftarrow A(x)$ for all concept names~$A$ such that $\T, \{A(a)\} \models \q_0$. 

The program $\Pi''_{\omq}$ is inspired by a similar construction from~\cite{LICS14}. By adapting results from the latter paper,
we can show that $(\Pi''_{\omq}, P_{\q_0}(\vec{x}))$ is indeed a rewriting: 
\begin{lemma}\label{bl-logcfl-rewriting-correct}
For any tree-shaped OMQ $\omq(\vec{x})=(\T, \q_0(\vec{x}))$, any $\q(\vec{z}) \in \sqset$, any H-complete ABox $\A$, and any tuple $\vec{a}$ in $\ind(\mathcal{A})$, 
$\Pi''_{\omq},\A \models P_\q(\vec{a})$ iff there exists a homomorphism $h\colon \q \to \can$ such that $h(\vec{z})= \vec{a}$. 
\end{lemma}

Now fix $\ell > 1$, and consider the class of OMQs $\omq(\vec{x})=(\T, \q(\vec{x}))$ with tree-shaped $\q(\vec{x})$ having at most $\ell$ leaves. 
The size of $\Pi''_{\omq}$ is polynomially bounded in $|\omq|$, since bounded-leaf CQs have polynomially many tree witnesses 
and also polynomially many tree-shaped subCQs. 
It is readily seen that 
the function $\nu$ defined by setting $\nu (P_{\q'}) = |\q'|$
is a weight function for $(\Pi''_{\omq},P_\q)$ such that $\nu (P_{\q}) \leq |\omq|$. 
Moreover, by Lemma~\ref{PrepMiddleVertex}, $\dep(\Pi,G) \le \log \nu(P_{\q}) +1$. 
We can thus apply Corollary~\ref{cor:weight} to conclude that the obtained NDL-rewritings can be 
evaluated in \LOGCFL. Finally, we note that since the number of leaves is bounded, 
it is in $\NL$ to decide whether a vertex satisfies the conditions of Lemma~\ref{PrepMiddleVertex},
and it is in \LOGCFL{} to decide whether $\T, \{A(a)\} \models \q_0$ \cite{DBLP:conf/lics/BienvenuKP15}
or whether a (logspace) representation of a possible tree witness is 
indeed a tree witness. This allows us to show that $(\Pi''_{\omq},P_\q)$ can be generated by an $\mathsf{L}^\LOGCFL$-transducer.


\section{Conclusions}

As shown above, for three important classes of OMQs, NDL-rewritings can be constructed and evaluated by theoretically optimal \NL{} and \LOGCFL{} algorithms. To see whether these rewritings are viable in practice, we generated three sequences of OMQs with the  ontology from Example~\ref{ex:rewriting:2} and linear CQs of up to 15 atoms as in Example~\ref{ex:rewriting:1}. We compared our \NL{} and \LOGCFL{} rewritings from Secs.~\ref{sec:5} and~\ref{sec:boundedtw}  (called \textsc{Lin} and \textsc{Log})\linebreak  with those produced by Clipper~\cite{DBLP:conf/aaai/EiterOSTX12} and Rapid~\cite{DBLP:conf/cade/ChortarasTS11}. The barcharts below show the \mbox{number} of clauses in the rewritings over H-complete ABoxes. While \textsc{Lin} and \textsc{Log} grow linearly\linebreak (in accord with theory), Clipper and Rapid failed to produce rewritings for longer CQs.\\[-1pt]
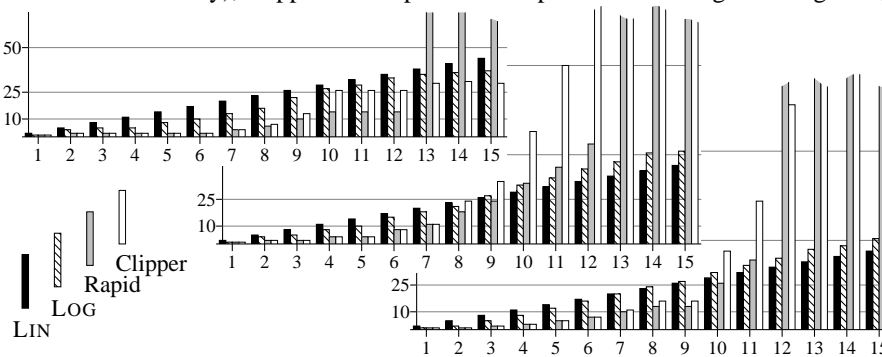
\begin{tikzpicture}[xscale=0.43,yscale=0.95]
\def\ysc{0.025}
\def\cwidth{0.2}
%
\begin{scope}
\clip (-0.5,-1) rectangle (16,145*\ysc);
\fill[white] (-0.1,-0.4) rectangle (15.5,125*\ysc);
\foreach \x in {1,...,15} {
\draw (\x,-0.1) -- (\x,0);
\node at (\x,-0.25) {\scriptsize \x};
}
\draw (0.7,0) -- ++(0,150*\ysc);
\foreach \y in {10,25,50,100} {
\draw (0.7,\y*\ysc) -- +(-0.2,0);
\draw[ultra thin,gray] (0.7,\y*\ysc) -- +(14.8,0);
\node at (0.2,\y*\ysc) {\scriptsize \y};
}
\foreach \x/\v in {1/2,2/5,3/8,4/11,5/14,6/17,7/20,8/23,9/26,10/29,11/32,12/35,13/38,14/41,15/44} { 
\filldraw (\x-2*\cwidth,0) -- ++(0,\v*\ysc) -- ++(\cwidth,0) -- (\x-\cwidth,0);
}
\foreach \x/\v in {1/1,2/2,3/5,4/8,5/12,6/16,7/20,8/24,9/27,10/32,11/36,12/40,13/45,14/47,15/51} {
\filldraw[pattern=north west lines, pattern color=black] (\x-\cwidth,0) -- ++(0,\v*\ysc) -- ++(\cwidth,0) -- (\x,0);
}
\foreach \x/\v in {1/1,2/1,3/2,4/3,5/5,6/7,7/10,8/13,9/13,10/26,11/39,12/250,13/250,14/250,15/250} {
\filldraw[fill=gray!50,ultra thin] (\x,0) -- ++(0,\v*\ysc) -- ++(\cwidth,0) -- (\x+\cwidth,0);
}
\foreach \x/\v in {1/1,2/1,3/2,4/3,5/5,6/7,7/11,8/16,9/16,10/44,11/72,12/126,13/241,14/250,15/250} {
\filldraw[fill=white,ultra thin] (\x+\cwidth,0) -- ++(0,\v*\ysc) -- ++(\cwidth,0) -- (\x+2*\cwidth,0);
}
\filldraw[white,draw=white,decorate, decoration={coil,segment length=20pt,aspect=0}] (0,140*\ysc) rectangle ++(17,1);
\draw (0.5,0) -- (15.5,0);
\end{scope}
%
%
%
\begin{scope}[xshift=-60mm,yshift=12mm]
\clip (-0.5,-1) rectangle (16,135*\ysc);
\fill[white] (-0.1,-0.4) rectangle (15.5,130*\ysc);
\foreach \x in {1,...,15} {
\draw (\x,-0.1) -- (\x,0);
\node at (\x,-0.25) {\scriptsize \x};
}
\draw (0.7,0) -- ++(0,100*\ysc);
\foreach \y in {10,25,50,100} {
\draw (0.7,\y*\ysc) -- +(-0.2,0);
\draw[ultra thin,gray] (0.7,\y*\ysc) -- +(14.8,0);
\node at (0.2,\y*\ysc) {\scriptsize \y};
}
\foreach \x/\v in {1/2,2/5,3/8,4/11,5/14,6/17,7/20,8/23,9/26,10/29,11/32,12/35,13/38,14/41,15/44} { 
\filldraw (\x-2*\cwidth,0) -- ++(0,\v*\ysc) -- ++(\cwidth,0) -- (\x-\cwidth,0);
}
\foreach \x/\v in {1/1,2/4,3/5,4/8,5/10,6/15,7/18,8/21,9/27,10/33,11/37,12/42,13/46,14/51,15/52} {
\filldraw[pattern=north west lines, pattern color=black] (\x-\cwidth,0) -- ++(0,\v*\ysc) -- ++(\cwidth,0) -- (\x,0);
}
\foreach \x/\v in {1/1,2/2,3/2,4/4,5/4,6/8,7/11,8/18,9/24,10/34,11/43,12/56,13/250,14/250,15/250} {
\filldraw[fill=gray!50,ultra thin] (\x,0) -- ++(0,\v*\ysc) -- ++(\cwidth,0) -- (\x+\cwidth,0);
}
\foreach \x/\v in {1/1,2/2,3/2,4/4,5/4,6/8,7/11,8/24,9/35,10/63,11/100,12/302,13/250,14/250,15/250} {
\filldraw[fill=white,ultra thin] (\x+\cwidth,0) -- ++(0,\v*\ysc) -- ++(\cwidth,0) -- (\x+2*\cwidth,0);
}
\draw (0.5,0) -- (15.5,0);
\filldraw[white,draw=white,decorate, decoration={coil,segment length=20pt,aspect=0}] (8,130*\ysc) rectangle ++(8,1);
\end{scope}
%
\begin{scope}[xshift=-120mm,yshift=27mm]
\clip (-0.5,-1) rectangle (16,75*\ysc);
\fill[white] (-0.1,-0.4) rectangle (15.5,55*\ysc);
\foreach \x in {1,...,15} {
\draw (\x,-0.1) -- (\x,0);
\node at (\x,-0.25) {\scriptsize \x};
}
\draw (0.7,0) -- ++(0,70*\ysc);
\foreach \y in {10,25,50,100} {
\draw (0.7,\y*\ysc) -- +(-0.2,0);
\draw[ultra thin,gray] (0.7,\y*\ysc) -- +(14.8,0);
\node at (0.2,\y*\ysc) {\scriptsize \y};
}
\foreach \x/\v in {1/2,2/5,3/8,4/11,5/14,6/17,7/20,8/23,9/26,10/29,11/32,12/35,13/38,14/41,15/44} { 
\filldraw (\x-2*\cwidth,0) -- ++(0,\v*\ysc) -- ++(\cwidth,0) -- (\x-\cwidth,0);
}
\foreach \x/\v in {1/1,2/4,3/5,4/5,5/8,6/10,7/13,8/16,9/22,10/27,11/29,12/33,13/35,14/36,15/37} {
\filldraw[pattern=north west lines, pattern color=black] (\x-\cwidth,0) -- ++(0,\v*\ysc) -- ++(\cwidth,0) -- (\x,0);
}
%
\foreach \x/\v in {1/1,2/2,3/2,4/2,5/2,6/2,7/4,8/6,9/10,10/14,11/14,12/14,13/250,14/250,15/250} {
\filldraw[fill=gray!50,ultra thin] (\x,0) -- ++(0,\v*\ysc) -- ++(\cwidth,0) -- (\x+\cwidth,0);
}
\foreach \x/\v in {1/1,2/2,3/2,4/2,5/2,6/2,7/4,8/7,9/13,10/26,11/26,12/26,13/30,14/31,15/30} {
\filldraw[fill=white,ultra thin] (\x+\cwidth,0) -- ++(0,\v*\ysc) -- ++(\cwidth,0) -- (\x+2*\cwidth,0);
}
\draw (0.5,0) -- (15.5,0);
\filldraw[white,draw=white,decorate, decoration={coil,segment length=20pt,aspect=0}] (13,70*\ysc) rectangle ++(3,1);
\end{scope}
\begin{scope}[xshift=-5mm,yshift=-2mm]
\filldraw (-11,0.5) rectangle +(\cwidth,30*\ysc);
\node at (-10.7,0.2) {\footnotesize \textsc{Lin}};
\filldraw[pattern=north west lines, pattern color=black] (-10,0.8) rectangle +(\cwidth,30*\ysc);
\node at (-9.4,0.5) {\footnotesize \textsc{Log}};
\filldraw[fill=gray!50,ultra thin]  (-9,1.1) rectangle +(\cwidth,30*\ysc);
\node at (-8.2,0.8) {\footnotesize Rapid};
\filldraw[fill=white,ultra thin]  (-8,1.4) rectangle +(\cwidth,30*\ysc);
\node at (-7,1.1) {\footnotesize Clipper};
\end{scope}
\end{tikzpicture}\\[-18pt]
We evaluated the rewritings over a few randomly generated ABoxes using off-the-shelf datalog engine RDFox~\cite{DBLP:conf/semweb/NenovPMHWB15}. The experiments (see the full version) show that our rewritings are usually executed  faster than Clipper's and Rapid's  when the number of answers is relatively small ($\lesssim 10^4$); for queries with $\gtrsim 10^6$ answers, the execution times are comparable. The version of RDFox we used did not seem to take advantage of the structure of the \NL{}/\LOGCFL{} rewritings, and it would be interesting to see whether their nonrecursiveness and  parallelisability can be utilised to produce efficient execution plans.

\nopagebreak[4]

\smallskip

\noindent\textbf{Acknowledgements\textup{:}} The first author was supported by contract ANR-12-JS02-007-01, the fourth by the Russian Foundation for Basic Research and the grant MK-7312.2016.1.

\pagebreak

\appendix

\newpage

\section{Proofs}

\noindent\textbf{Lemma \ref{linear-arbitrary}}.  \ \
{\it For any fixed $\wid > 0$, there is an $\mathsf{L}^\NL$-transducer that, given a linear \NDL-rewriting of an OMQ $\omq(\vec{x})$ over H-complete ABoxes that is of width at most $\wid$, computes a linear NDL-rewriting of $\omq(\vec{x})$ over arbitrary ABoxes whose width is at most $\wid+1$.}
\begin{proof}
Let $(\Pi, G(\vec{x}))$ be a linear \NDL-rewriting of the OMQ $\omq(\vec{x}) = (\T,\q(\vec{x}))$ over H-complete ABoxes of width $\wid$. 
and we will replace every clause $\lambda$ in $\Pi$ by a set of clauses $\lambda^*$ defined as follows. 
Suppose $\lambda$ is of the form
\begin{equation*}
H(\vec{z}) \gets I \land \textit{EQ} \land E_1 \land \ldots \land E_n,
\end{equation*}
where $I$ is the only IDB body atom in $\lambda$, $\textit{EQ}$ contains all equality body atoms, and $E_1, \ldots, E_n$ are the EDB body atoms not involving equality. 
For every atom $E_i$, we define a set $\upsilon(E_i)$ of atoms by taking 
\begin{align*}
\upsilon(E_i) & =
\bigl\{B(u) \mid B \sqsubseteq_\T A\bigr\} \cup \bigl\{R(u,u_i) \mid \exists R \sqsubseteq_\T A\bigr\},  &&\text{ if } E_i= A(u),\\
\upsilon(E_i) & = \bigl\{R'(u,v) \mid R' \sqsubseteq_\T R\bigr\},  &&\text{ if } E_i = R(u,v),
\end{align*}
where $u_i$ is a fresh variable not occurring in $\lambda$; 
we assume $P^-(u,v)$ coincides with $P(v,u)$, for all role names $P$.
Intuitively, $\upsilon(E_i)$ captures all atoms that imply $E_i$ with respect to $\T$. 
Then $\lambda^*$ consists of the following clauses: 
\begin{align*}
H_0(\vec{z}_0) & \gets I,\\
H_{i+1}(\vec{z}_i) & \gets H_i(\vec{z}_i) \land E_i', && \text{ for every } 1 \leq i \leq n \text{ and every }E_i' \in \upsilon(E_i),\\
H(\vec{z}) & \gets H_{n+1}(\vec{z}_n) \land \textit{EQ},
\end{align*}
where $\vec{z}_i$ is the restriction of $\vec{z}$ to variables occurring in  $I$ if $i = 0$ and in  $H_i(\vec{z}_i)$ and $E_i'$ except for $u_i$ if $i > 0$ (note that $\vec{z}_n = \vec{z}$).
%
Let $\Pi'$ be the program obtained from $\Pi$ by replacing each clause $\lambda$ by the set of clauses~$\lambda^*$. 
By construction, $\Pi'$ is a linear NDL program and its width cannot exceed $\wid+1$ (the possible increase of $1$ is due to  
the replacement of concept atoms by role atoms). 

We now argue that $(\Pi', G(\vec{x}))$ is a rewriting of $\omq(\vec{x})$ over arbitrary ABoxes. 
It is easily verified that $(\Pi', G(\vec{x}))$ is equivalent to $(\Pi'', G(\vec{x}))$, where 
$\Pi''$ is obtained from $\Pi$ by replacing each clause 
$H(\vec{z}) \gets I \land \textit{EQ} \land E_1 \land \ldots \land E_n$ by the (possibly exponentially larger) set of clauses 
\begin{equation*}
\bigl\{H(\vec{z}) \gets I \land \textit{EQ} \land E_1' \land \ldots \land E_n' \mid E_i' \in \upsilon(E_i), 1 \leq i \leq n\bigr\}.
\end{equation*}
It thus suffices to show that  $(\Pi'', G(\vec{x}))$ is a rewriting of $\omq(\vec{x})$ over arbitrary ABoxes. 

First suppose that $\T, \A \models \q(\vec{a})$, where $\A$ is an arbitrary ABox. 
Let $\A'$ be the H-complete ABox obtained from $\A$ by adding the assertions:
\begin{itemize}
\item $P(a,b)$ whenever $R(a,b) \in \A$ and $R \sqsubseteq_\T P$;
\item $A(a)$ whenever $B(a) \in \A$ (with $B$ a basic concept) and $B \sqsubseteq_\T A$.
\end{itemize}
Clearly, $\T, \A' \models \q(\vec{a})$, so we must have $\Pi, \A' \models G(\vec{a})$. 
A simple inductive argument (on the order of derivation of ground atoms) 
shows that whenever a clause \mbox{$H(\vec{z}) \gets I \land \textit{EQ}\land E_1 \land \ldots \land E_n$}
is applied using a substitution $\vec{c}$ for the\linebreak variables in the body to derive $H(\vec{c}(\vec{z}))$ using $\Pi$,
we can find a corresponding clause\linebreak \mbox{$H(\vec{z}) \gets I \land \textit{EQ}\land E_1' \land \ldots \land E_n'$} and a substitution $\vec{c}'$ extending $\vec{c}$ (on the fresh variables $u_i$)
that allows us to derive $H(\vec{c}'(\vec{z}))$ using $\Pi''$.
Indeed, 
if $E_i=A(u)$, then $A(\vec{c}(u)) \in \A'$, so there must exist either a concept assertion $A'(\vec{c}(u)) \in \A$ such that $A' \sqsubseteq_\T A$
or a role assertion $R(a,b) \in \A$ such that $\exists R \sqsubseteq_\T A$. 
Similarly, if $E_i=R(u,v)$, then there must exist a role assertion $R'(u,v) \in \A$ such that $R' \sqsubseteq_\T R$. 
It then suffices to choose a clause  $H(\vec{z}) \gets I \land \textit{EQ} \land E_1' \land \ldots \land E_n'$
with atoms $E_i'$ whose form matches that of the assertion in $\A$ corresponding to $E_i$.  

For the converse direction, it suffices to observe that $\Pi\subseteq \Pi''$. 


To complete the proof, we note that it is in \NL\ to decide whether an atom belongs to $\upsilon(E_i)$, 
and thus we can construct the program $\Pi'$ by means of an $\mathsf{L}^\NL$-transducer. \qed
\end{proof}

\bigskip

\noindent\textbf{Lemma~\ref{prop:logdepth:rewriting}.} \ \ {\itshape
For any ABox $\A$, any $D \in \R$, any type $\tpd$ with $\dom(\tpd) = \dD$, any $\vec{b} \in \ind(\A)^{|\dD|}$ and $\vec{a}\in \ind(\A)^{|\vec{x}_D|}$, we have $\Pi_\omq,\A \models \rpred^{\tpd}_D(\vec{b}, \vec{a})$  iff there is a homomorphism $h\colon \q_D \to \can$ such that  
\begin{equation}\label{eq3}
h(x) = \vec{a}(x), \text{ for } x\in \vec{x}_D, \quad\text{ and }\quad 
h(v) = \vec{b}(v) \tpd(v), \text{ for  } v\in \dD.
\end{equation}
}
\begin{proof}
$(\Rightarrow)$
The proof is by induction on $\prec$. For the base of induction, let $D$ be of size~1. 
By the definition of $\Pi_\omq$, there exists
a type $\tpr$  such that  $\dom(\tpr) = \lambda(\sigma(D))$ and $\tpd$ agrees with $\tpr$ on $\dD$ and a respective tuple 
$\vec{c} \in \ind(\A)^{|\lambda(\sigma(D))|}$ such that $\vec{c}(v) = \vec{b}(v)$, 
for all $v \in \dD$,  and $\vec{c}(x) = \vec{a}(x)$, for all $x\in \vec{x}_D$, and
$\Pi_\omq,\A \models \mathsf{At}^{\tpr}(\vec{c})$. Then, for any atom $S(\vec{v}) \in \q_D$, we have $\vec{v}\subseteq\lambda(\sigma(D))$, whence $\can \models S(h(\vec{v}))$ as $\vec{w}$ agrees with $\vec{s}$ on $\dD$.

For the inductive step, suppose that $\Pi_\omq,\A\models \rpred^{\tpd}_D(\vec{b}, \vec{a})$. By the definition of $\Pi_\omq$, there exists
a type $\tpr$  such that  $\dom(\tpr) = \lambda(\sigma(D))$ and $\tpd$ agrees with $\tpr$ on their common domain and a respective tuple 
$\vec{c} \in \ind(\A)^{|\lambda(\sigma(D))|}$ such that $\vec{c}(v) = \vec{b}(v)$, 
for all $v \in \dD$,  and $\vec{c}(x) = \vec{a}(x)$, for all $x\in \vec{x}_D$, and
\begin{equation*}
\Pi_\omq,\A \models \mathsf{At}^{\tpr}(\vec{c}) \land 
\bigwedge\nolimits_{D' \prec D} \rpred^{(\tpr\cup\tpd) \restr\dDp}_{D'}(\vec{b}_{D'},\vec{a}_{D'}),
\end{equation*}
where $\vec{b}_{D'}$ and $\vec{a}_{D'}$ are the restrictions of 
$\vec{b} \cup \vec{c}$ to $\dDp$ and  of $\vec{a}$ to 
$\vec{x}_{D'}$, respectively. 
By the induction hypothesis, for any $D'\prec D$, there is a  homomorphism $h_{D'}\colon\q_{D'}\to\can$ such that~\eqref{eq3} is satisfied.
%

Let us show that the $h_{D'}$ agree on common variables. Suppose that $v$ is shared by $\q_{D'}$ and
$\q_{D''}$ for $D' \prec D$ and $D'' \prec D$.  By 
the definition of tree decomposition,  
for every $v \in V$, the nodes $\{\nd\mid v \in \lambda(\nd)\} $ induce a connected subtree of~$T$, and so 
$v \in \lambda(\sigma(D)) \cap \lambda(\nd') \cap \lambda(\nd'')$, where
$\nd'$ and $\nd''$ are the unique neighbours of $\sigma(D)$ lying in $D'$ and $D''$, respectively.
%
%
Since  
$\tpd'=(\tpd\cup\tpr) \restr\dDp$ and $\tpd''=(\tpd\cup\tpr) \restr\dDpp$ are the restrictions of 
$\tpd \cup\tpr$, we have
$\tpd'(v)  = \tpd''(v)$.
This implies that $h_{D'}(v)  = \vec{c}(v) \vec{w}'(v) = \vec{c}(v) \vec{w}''(v) = h_{D''}(v)$.

Now we define $h$ on every $v$ in $\q_D$ by taking
\begin{equation*}
h(v) = \begin{cases}
h_{D'}(v) & \text{ if }  v \in\lambda(t), \text{ for }t\in D' \text{ and } D' \prec D,\\
\vec{c}(v)\cdot (\tpd\cup\tpr)(v), & \text{ if } v \in\lambda(\sigma(D)).	
\end{cases}
\end{equation*}
If follows that $h$ is well defined, $h$ satisfies (\ref{eq3}) and that $h$ is a homomorphism
from $\q_D$ to $\can$. 
Indeed, take an atom $S(\vec{v}) \in \q_D$. Then either $\vec{v}\subseteq\lambda(\sigma(D))$,  
in which case $\can \models S(h(\vec{v}))$ since $\vec{w}$ is compatible
with $\sigma(D)$ and $\Pi_\omq, \A \models \mathsf{At}^{\tpr}(\vec{c})$, or
\mbox{$S(\vec{v}) \in \q_{D'}$} for some $D' \prec D$, in which case we use the fact
that $h$ extends a homomorphism $h_{D'}$.

\medskip

$(\Leftarrow)$
The proof is by induction on $\prec$. Fix $D$ and $\tpd$ such that $|\tpd| = |\dD|$.
Take $\vec{b} \in \ind(\A)^{|\dD|}$, $\vec{a} \in \ind(\A)^{|\vec{x}_D|}$, and
a homomorphism $h\colon\q_D\to\can$ satisfying~\eqref{eq3}. 
Define a type $\tpr$ and a tuple $\vec{c} \in \ind(\A)^{|\lambda(\sigma(D))|}$ by taking, for all $v\in \lambda(\sigma(D))$,
\begin{equation*}
\tpr(v) = w \ \ \text{ and } \ \  \vec{c}(v) = a, \ \ \ \ \ \text{ if } h(v) = a w, \text{ for  } a\in\ind(\A).
\end{equation*}
By definition, $\dom(\tpr) = \lambda(\sigma(D))$ and, by~\eqref{eq3}, $\tpr$ and $\tpd$ agree on the common domain.
For the inductive step, for each $D'\prec D$,  let
$h_{D'}$ be the restriction of $h$ to $\q_{D'}$ and let $\vec{b}_{D'}$ and and $\vec{a}_{D'}$  be the restrictions of 
$\vec{b} \cup \vec{c}$ to $\dDp$ and  of $\vec{a}$ to 
$\vec{x}_{D'}$, respectively. 
By the inductive hypothesis, $\Pi_\omq, \A \models \rpred^{\tpd'}_{D'}( \vec{b}_{D'},\vec{a}_{D'})$. (This argument is not needed for the basis of induction.)
Since $h$ is a homomorphism, we have
$\Pi_\omq,\A \models \mathsf{At}^{\tpr}(\vec{c})$,
whence, $\Pi_\omq,\A\models  \rpred^{\tpd}_D(\vec{b}, \vec{a})$.
\qed\end{proof}

\bigskip

\noindent\textbf{Lemma \ref{nl-rewriting-correct}.} \ \
{\it For any H-complete ABox $\A$, any $0 \leq n \leq M$, any predicate $P^{\vec{w}}_n$, 
any $\vec{b} \in \ind(\A)^{|\vec{z}^n_\exists|}$ and any $\vec{a}\in \ind(\A)^{|\vec{x}^n|}$,
we have  
$\Pi_\omq',\A \models P^{\vec{w}}_n(\vec{b}, \vec{a})$  iff there is a homomorphism 
$h\colon \q_n \to \can$ such that  
\begin{equation}\tag{\ref{nl-rewriting-eq}}
h(x) = \vec{a}(x), \text{ for } x\in \vec{x}^n, \quad\text{ and }\quad 
h(z) = \vec{b}(z) \vec{w}(z), \text{ for  } z\in\vec{z}^n_\exs.
\end{equation}
}
\begin{proof}
The proof is by induction on $n$. For the base case ($n=M$), first suppose that we have
$\Pi_\omq',\A \models P^{\vec{w}}_M(\vec{b}, \vec{a})$. 
The only rule in $\Pi_\omq'$ with head predicate $P^{\vec{w}}_M$ is $P^{\vec{w}}_{M}(\vec{z}^{M}_\exs, \vec{x}^M) \leftarrow \mathsf{At}^{\vec{w}}(\vec{z}^M)$ with $\vec{z}^M = \vec{z}^M_\exs \uplus \vec{x}^M$, which is equivalent to
\begin{equation}\label{rule-M}
P^{\vec{w}}_{M}(\vec{z}^{M}_\exs, \vec{x}^M) \leftarrow \bigwedge_{z \in \vec{z}^M} \Bigl(\bigwedge_{\substack{A(z) \in \q\\ \vec{w}(z) = \varepsilon}} \!\!\!\!A(z)
\,\, \land \!\! \bigwedge_{\substack{R(z, z) \in \q\\ \vec{w}(z) = \varepsilon}} \!\!\!\! \!\! R(z, z) \,\,\land \!\!\! 
\bigwedge_{\substack{\vec{w}(z) = S w'}} \!\!\!A_S(z)\Bigr).
\end{equation}
So the body of this rule must be satisfied when $\vec{b}$ and $\vec{a}$ are substituted for $\vec{z}^M_\exs$ and $\vec{x}^M$ respectively. 
Moreover, by local compatibility of $\vec{w}$ with $\vec{z}^M$, we know that $\vec{w}(x) = \varepsilon$ for every $x \in \vec{x}^M$. 
It follows that
\begin{itemize}
\item[--] $A(\vec{a}(x)) \in \A$ for every $A(x) \in \q$ such that $x \in \vec{x}^M$; 
\item[--] $A(\vec{b}(z)) \in \A$ for every $A(z) \in \q$ such that $z \in \vec{z}^M_\exs$ and $\vec{w}(z) = \varepsilon$;
\item[--] $R(\vec{a}(x),\vec{a}(x)) \in \A$ for every $R(x, x) \in \q$ such that $x \in \vec{x}^M$;
\item[--] $R(\vec{b}(z),\vec{b}(z)) \in \A$ for every $R(z, z) \in \q$ such that $z \in \vec{z}^M_\exs$ and $\vec{w}(z) = \varepsilon$;
\item[--] $A_S(z) \in \A$ for every $z\in \vec{z}^M$ with $\vec{w}(z) = S w'$.
\end{itemize}
Now let $h^M$ be the unique mapping from $\vec{z}^M$ to $\Delta^{\can}$ satisfying~\eqref{nl-rewriting-eq}. First note that $h^M$ is well-defined, since by the last item, whenever $\vec{w}(z) = S w'$,
we have $A_S(z) \in \A$ and $S w' \in \twords$, so $\vec{b}(z) S w'$ belongs to $\Delta^{\can}$. 
To show that $h^M$ is a homomorphism of $\q_M$ into $\can$, first recall that the atoms of $\q_M$ are of two types: $A(z)$ or $R(z, z)$, 
with $z \in \vec{z}^M$. 
Take some $A(z) \in \q_M$. If $\vec{w}(z) = \varepsilon$, then we immediately obtain either 
$A(h^M(z))=A(\vec{a}(z)) \in \A$ or  $A(h^M(z))=A(\vec{b}(z)) \in \A$, depending on whether $z \in \vec{z}^M_\exs$ or in $\vec{x}^M$. 
Otherwise, if $\vec{w}(z) \neq \varepsilon$, then the local compatibility of $\vec{w}$ with $\vec{z}^M$ means that the final letter $R$ in $\vec{w}(z)$ is such that $\exists R^- \sqsubseteq_\T A$, hence $h^M(z) = \vec{b}(z) \vec{w}(z) \in A^{\can}$. 
Finally, suppose that $R(z, z) \in \q$. The local compatibility of $\vec{w}$ with $\vec{z}^M$ ensures that $\vec{w}(z) = \varepsilon$, and thus 
we have either $R(\vec{a}(z),\vec{a}(z)) \in \A$ or $R(\vec{b}(z),\vec{b}(z)) \in \A$, depending again on whether  
$z \in \vec{z}^M_\exs$ or $z \in\vec{x}^M$.

For the other direction of the base case, suppose that the mapping $h^M$ given by~\eqref{nl-rewriting-eq} defines a homomorphism from $\q_M$ into $\can$. 
We therefore have:
\begin{itemize}
\item[--] $\vec{a}(x) \in A^{\can}$ for every $A(x) \in \q$ with $x \in \vec{x}^M$;
\item[--] $\vec{b}(z) \vec{w}(z) \in A^{\can}$ for every $A(z) \in \q$ with $z \in \vec{z}^M_\exs$;
\item[--] $(\vec{a}(x),\vec{a}(x)) \in R^{\can}$ for every $R(x, x) \in \q$ such that $x \in \vec{x}^M$;
\item[--] $(\vec{b}(z),\vec{b}(z)) \in R^{\can}$ for every $R(z, z) \in \q$ such that $z \in \vec{z}^M_\exs$; 
\item[--]  $\T, \A \models \exists S(\vec{b}(z))$ for every $z \in \vec{z}^M_\exs$  with $\vec{w}(z)=S w'$ (for otherwise $\vec{b}(z)\vec{w}(z)$
would not belong to the domain of $\can$).
\end{itemize}
The first two items, together with H-completeness of the ABox $\A$, ensure that all atoms in $\{A(z) \mid A(z) \in \q, z\in\vec{z}^M, \vec{w}(z) = \varepsilon\}$ 
are present in $\A$ when $\vec{b}$ and $\vec{a}$  substituted for $\vec{z}^M_\exs$ and $\vec{x}^M$ respectively. 
The third and fourth items, again together with H-completeness, ensure the presence of the atoms in $\{R(z, z) \mid R(z, z) \in \q, z\in\vec{z}^M, \vec{w}(z) = \varepsilon\}$.
Finally, the fifth item plus H-completeness ensures that $\A$ contains all atoms in $\{A_S(z) \mid z\in\vec{z}^M, \vec{w}(z) = S w'\}$. 
It follows that the body of the unique rule for $P^{\vec{w}}_{M}$
is satisfied when $\vec{b}$ and $\vec{a}$ are substituted for $\vec{z}^M_\exs$ and $\vec{x}^M$ respectively, and thus
$\Pi_\omq',\A \models P^{\vec{w}}_M(\vec{b}, \vec{a})$. 

\bigskip

For the induction step, assume that the statement has been shown to hold for all $n \leq k+1 \leq M$,
and let us show that it holds when $n=k$. For the first direction, suppose $\Pi_\omq',\A \models P^{\vec{w}}_k(\vec{b}, \vec{a})$. 
It follows that there exists a pair of types $(\vec{w}, \tpr)$ compatible with $(\vec{z}^k, \vec{z}^{k+1})$
and an assignment $\vec{c}$ of individuals from $\A$ to the variables in $\vec{z}^k \cup \vec{z}^{k+1}$
such that $\vec{c}(x)=\vec{a}(x)$ 
for all $x \in (\vec{z}^k \cup \vec{z}^{k+1})\cap \vec{x}$, and $\vec{c}(z)=\vec{b}(z)$ for all $z \in \vec{z}^k_\exs$, and 
such that every atom in the body of the clause
\begin{equation*}
P^{\vec{w}}_{k}(\vec{z}^{k}_\exs, \vec{x}^k) \leftarrow 
\mathsf{At}^{\vec{w}\cup\tpr}(\vec{z}^k,\vec{z}^{k+1})  \land P^{\tpr}_{k+1}(\vec{z}^{k+1}_\exs, \vec{x}^{k+1})
\end{equation*}
is entailed from $\Pi_\omq',\A$ when the individuals in $\vec{c}$ are substituted for $\vec{z}^k \cup \vec{z}^{k+1}$. We
recall that  $\mathsf{At}^{\vec{w}\cup\tpr}(\vec{z}^k,\vec{z}^{k+1})$ is the conjunction of the following atoms, for $z,z'\in \vec{z}^k\cup\vec{z}^{k+1}$:
\begin{itemize}
\item[--] $A(z)$, if $A(z) \in \q$ and $(\vec{w}\cup\tpr)(z) = \varepsilon$,
\item[--] $R(z, z')$,  if $R(z, z') \in \q$ and $(\vec{w}\cup\tpr)(z) = (\vec{w}\cup\tpr)(z') = \varepsilon$,
\item[--] $z = z'$, if $R(z, z') \in \q$ and either $(\vec{w}\cup\tpr)(z) \neq \varepsilon$  or $(\vec{w}\cup\tpr)(z') \neq \varepsilon$,
\item[--] $A_S(z)$, if $(\vec{w}\cup\tpr)(z)$ is of the form $S w'$. 
\end{itemize}
In particular, we have $\Pi_\omq',\A \models  P^{\tpr}_{k+1}(\vec{c}(\vec{z}^{k+1}_\exs), \vec{c}({\vec{x}^{k+1}}))$. 
By the induction hypothesis, there exists a homomorphism $h^{k+1}\colon \q_{k+1} \to \can$ such that
$h^{k+1}(u) = \vec{c}(u) \vec{s}(u)$ for every $u \in \vec{z}^{k+1}_\exs \cup \vec{x}^{k+1}$. 
Define a mapping $h^k$ from $\vars(\q_k)$ to $\Delta^{\can}$ by setting $h^k(u)=h^{k+1}(u)$ for every variable $u\in \vars(\q_{k+1})$,
setting $h^k(x) = \vec{a}(x)$ for every $x \in \vec{z}^k \cap \vec{x}$, and setting 
$h^k(z)=\vec{b}(z) \vec{w}(z)$ for every $z \in \vec{z}^k$. 
Using the same argument as was used in the base case, 
we can show that $h^k$ is well-defined. For atoms from $\q_k$ involving only variables from $\q_{k+1}$, we can 
use the induction hypothesis to conclude that they are satisfied under $h^k$, 
and for atoms only involving variables from $\vec{z}^k$, we can argue as in the base case. 
It thus remains to handle role atoms that contain one variable from $\vec{z}^k$ and one variable from $\vec{z}^{k+1}$. 
Consider such an atom $R(z, z') \in \q_k$, for $z\in\vec{z}^k$ and $z'\in\vec{z}^{k+1}$.
If $\vec{w}(z) = \tpr(z') = \varepsilon$, then the atom $R(z, z')$ appears in the body of the clause we are considering.
It follows that $\Pi_\omq',\A \models  R(\vec{c}(z), \vec{c}(z'))$,
hence $(\vec{c}(z), \vec{c}(z')) \in R^{\can}$. 
It then suffices to note that $\vec{c}$ agrees with $\vec{a}$ and $\vec{b}$ on the variables in $\vec{z}^k$.  
Next suppose that either $\vec{w}(z) \neq \varepsilon$ or $ \tpr(z') \neq \varepsilon$. 
It follows that the clause body contains $z = z'$, hence $\vec{c}(z) = \vec{c}(z')$. 
As $(\tpd, \tpr)$ is compatible with $(\vec{z}^k, \vec{z}^{k+1})$, one of the following must hold:
(a) $\tpr(z')= \tpd(z) R'$ with $ R' \sqsubseteq_\T R$, or
(b) $\tpd(z) = \tpr(z') R' $ with $R' \sqsubseteq_\T R^-$. 
We give the argument in the case where $z \in \vec{z}^k_\exs$ (the argument is entirely similar if $z \in \vec{x}^k$). 
If (a) holds, then 
$$
(h^k(z), h^k(z'))= (\vec{b}(z) \vec{w}(z), \vec{c}(z') \vec{s}(z'))= 
(\vec{b}(z) \vec{w}(z), \vec{c}(z') \tpd(z) R')\in R^{\can}
$$
since $ R' \sqsubseteq_\T R$ and $\vec{c}(z') =\vec{c}(z) = \vec{b}(z)$. 
If (b) holds, 
then 
$$
(h^k(z), h^k(z'))= (\vec{b}(z) \vec{w}(z), \vec{c}(z') \vec{s}(z'))= 
(\vec{b}(z) \vec{s}(z') R', \vec{c}(z') \vec{s}(z'))\in R^{\can}
$$
since $R' \sqsubseteq_\T R^-$. 

\smallskip

For the converse direction of the induction step, let $\vec{w}$ be a type that is locally compatible with $\vec{z}^k$,
 let  $\vec{a}\in \ind(\A)^{|\vec{x}^k|}$ and $\vec{b} \in \ind(\A)^{|\vec{z}^k_\exists|}$, and
let $h^k\colon \q_k \to \can$ be a homomorphism satisfying 
\begin{equation}\label{indstep-sec5}
h^k(x) = \vec{a}(x), \text{ for } x\in \vec{x}^k, \quad\text{ and }\quad 
h^k(z) = \vec{b}(z) \vec{w}(z), \text{ for  } z\in\vec{z}^k_\exs.
\end{equation}
We let $\vec{c}$ for $\vec{z}^{k+1}$ be defined by setting $\vec{c}(z)$ equal to the unique individual $c$ such that $h(z)$ is of the form $c w$ (for some $w \in \twords$),
and let $\vec{s}$ be the unique type for $\vec{z}^{k+1}$ satisfying $h(z) = \vec{c}(z) \vec{s}(z)$ for every $z \in \vec{z}^{k+1}$; in other words,  we obtain $\vec{s}(z)$ from $h(z)$ by omitting the initial individual name $\vec{c}(z)$. 
Note that since $\vec{x}^{k+1} \subseteq \vec{x}^{k}$, we have $\vec{a}(x)=\vec{c}(x)$ for every $x\in \vec{x}^{k+1}$.
It follows from the fact that $h^k$ is a homomorphism that $\vec{s}$ is locally compatible with $\vec{z}^{k+1}$
and that, for every role atom 
$R(z, z') \in \q_k$ with $z\in\vec{z}^k$ and $z'\in\vec{z}^{k+1}$, one of the following holds: 
(\emph{i}) $\tpd(z) =\tpr(z')= \varepsilon$, 
(\emph{ii}) $\tpr(z')= \tpd(z) R'$ with $R' \sqsubseteq_\T R$, or
(\emph{iii}) $\tpd(z) = \tpr(z') R'$ with $R' \sqsubseteq_\T R^-$.
Thus, the pair of types $(\vec{w}, \vec{s})$ is compatible with $(\vec{z}^{k}, \vec{z}^{k+1})$,
and so the following rule appears in $\Pi_\omq'$:
\begin{equation*}
P^{\vec{w}}_{k}(\vec{z}^{k}_\exs, \vec{x}^k) \leftarrow 
\mathsf{At}^{\vec{w}\cup\tpr}(\vec{z}^k,\vec{z}^{k+1})  \land P^{\tpr}_{k+1}(\vec{z}^{k+1}_\exs, \vec{x}^{k+1}),
\end{equation*}
where we recall that $\mathsf{At}^{\vec{w}\cup\tpr}(\vec{z}^k,\vec{z}^{k+1})$ is the conjunction of the following atoms, for $z,z'\in \vec{z}^k\cup\vec{z}^{k+1}$:
\begin{itemize}
\item[--] $A(z)$, if $A(z) \in \q$ and $(\vec{w}\cup\tpr)(z) = \varepsilon$,
\item[--] $R(z, z')$,  if $R(z, z') \in \q$ and $(\vec{w}\cup\tpr)(z) = (\vec{w}\cup\tpr)(z') = \varepsilon$,
\item[--] $z = z'$, if $R(z, z') \in \q$ and either $(\vec{w}\cup\tpr)(z) \neq \varepsilon$  or $(\vec{w}\cup\tpr)(z') \neq \varepsilon$,
\item[--] $A_S(z)$, if $(\vec{w}\cup\tpr)(z)$ is of the form $S w'$. 
\end{itemize}
It follows from Equation~\eqref{indstep-sec5} and the fact that $h^k$ is a homomorphism that each of the ground atoms obtained 
by taking an atom from $\mathsf{At}^{\vec{w}\cup\tpr}(\vec{z}^k,\vec{z}^{k+1})$ and substituting 
 $\vec{a}$, $\vec{b}$, and $\vec{c}$ for  $\vec{x}^k$, $\vec{z}^k_\exs$ and $\vec{z}^{k+1}$, respectively, is present in $\A$. 
By applying the induction hypothesis to the predicate $P^{\tpr}_{k+1}$ and the homomorphism $h^{k+1}\colon \q_{k+1} \to \can$
obtained by restricting $h^k$ to $\vars(\q_{k+1} )$, we obtain that  $\Pi_\omq',\A \models P^{\tpr}_{k+1}(\vec{c}(\vec{z}^{k+1}_\exs), \vec{a}(\vec{x}^{k+1}))$. 
Since for the considered substitution, all body atoms are entailed, we can conclude that \mbox{$\Pi_\omq',\A \models P^{\vec{w}}_{k}( \vec{b},\vec{a})$}.\qed
\end{proof}

\bigskip

\noindent\textbf{Lemma \ref{bl-logcfl-rewriting-correct}.} \ \
{\it For any tree-shaped OMQ $\omq(\vec{x})=(\T, \q_0(\vec{x}))$, any $\q(\vec{z}) \in \sqset$, any H-complete ABox $\A$, and any tuple $\vec{a}$ in $\ind(\mathcal{A})$, 
$\Pi''_{\omq},\A \models P_\q(\vec{a})$ iff there exists a homomorphism $h\colon \q \to \can$ such that $h(\vec{z})= \vec{a}$. 
}
\begin{proof}
An inspection of the definition of the set $\sqset$ shows that every $\q(\vec{z}) \in \sqset$ is a tree-shaped query having at least one answer variable,
with the possible exception of the original query $\q_0(\vec{x})$, which may be Boolean.

Just as we did for subtrees in Section~\ref{sec:boundedtw}, we associate a binary relation on the queries in $\sqset$ by setting 
$ \q'(\vec{z}') \prec \q(\vec{z})$ whenever $\q'(\vec{z}')$ was introduced when applying one of the two decomposition conditions on p.~\pageref{page:decomposition} to $\q(\vec{z})$. 
The proof is by induction on the subqueries in $\sqset$, according to $\prec$. 
We will start by establishing the statement 
for all queries in $\sqset$ other than $\q_0(\vec{x})$, 
and afterwards, we will complete the proof by giving an argument for $\q_0(\vec{x})$.

\smallskip

For the basis of induction, take some $\q(\vec{z}) \in \sqset$ that is minimal in the ordering induced by $\prec$, which means that
$\vars(\q) = \vec{z}$. Indeed, if there is an existentially quantified variable, then the first decomposition rule will give rise to a `smaller' query (in particular, if $|\vars(\q)| = 2$, then although the `smaller' query may have the same atoms, the selected existential variable will become an answer variable).
For the first direction, suppose that $\Pi''_{\omq},\A \models P_\q(a)$. By definition, $P_\q(z) \gets \q(z)$ is the only clause   with head predicate~$P_\q$. Thus, all atoms in the ground CQ $\q(\vec{a})$ are present in $\A$, and hence the desired homomorphism exists. 
For the converse direction, suppose there is a homomorphism $h\colon \q(\vec{z}) \to \can$ such that $h(\vec{z}) = \vec{a}$. 
It follows that every atom in the ground CQ $\q(\vec{a})$ is entailed from $\T, \A$. 
H-completeness of $\A$ ensures that all of the ground atoms in $\q(\vec{a})$ are present in $\A$, and thus we can apply the clause $P_\q(\vec{z}) \gets \q(\vec{z})$
to derive $P_\q(\vec{a})$. 

\smallskip

For the induction step, consider $\q(\vec{z}) \in \sqset$ with $\vars(\q)\ne\vec{z}$ and 
suppose that the claim holds for all $\q'(\vec{z}')  \in  \sqset$ with
$ \q'(\vec{z}') \prec \q(\vec{z})$. For the first direction, let $\Pi''_{\omq},\A \models P_\q(\vec{a})$.
There are two cases, depending on which type of clause was used to derive $P_\q(\vec{a})$.
\begin{itemize}
\item Case 1: $P_\q(\vec{a})$ was derived by an application of the following clause:
\begin{equation*}
\rew_\q(\vec{z}) \ \ \ \leftarrow \bigwedge_{A(v_\q)\in \q}\hspace*{-1em} A(v_\q) \ \ \ \land \bigwedge_{R(v_\q,v_\q) \in \q} \hspace*{-1em}R(v_\q,v_\q)  \ \ \  \land \bigwedge_{1\leq i\leq n} \rew_{\q_i}(\vec{z}_i),
\end{equation*}
where $\q_1(\vec{z}_1), \ldots, \q_n(\vec{z}_n)$ are the subqueries induced by the neighbours of $v_\q$ in the Gaifman graph $\gfmn$ of $\q$.
Then there exists a substitution $\vec{c}$ for the variables in the body of this rule that coincides with $\vec{a}$ on $\vec{z}$ and is 
such that the ground atoms obtained by applying $\vec{c}$ to the variables in the body are all entailed from $\Pi''_{\omq},\A$. 
In particular,  $\Pi''_{\omq},\A  \models \rew_{\q_i}(\vec{c}(\vec{z}_i))$
for every $1 \leq i \leq n$. We can apply the induction hypothesis to the $\q_i(\vec{z}_i)$ to obtain homomorphisms $h_i\colon \q_i \to \can$ such that  $h_i(\vec{z}_i)=\vec{c}(\vec{z}_i)$. 
Let $h$ be the mapping from $\vars(\q)$ to $\Delta^{\can}$ defined by taking 
$h(v)=h_i(v)$, for $v \in \vars(\q_i)$.  
Note that $h$ is well-defined since \mbox{$\vars(\q)=\bigcup_{i = 1}^n\vars(\q_i)$}, and 
the $\q_i$ have no variable in common other than $v_\q$, which is sent to $\vec{c}(v_\q)$ by every $h_i$. 
To see why $h$ is a homomorphism from $\q$ to $\can$, observe that \mbox{$\q = \bigcup_{i =1}^n \q_i \cup \{A(v_\q)\in \q\} \cup \{R(v_\q,v_\q) \in \q\}$}. 
By the definition of~$h$, all atoms in $\bigcup_{i = 1}^n \q_i$ hold under $h$. 
If $A(v_\q)\in \q$, then $A(\vec{c}(v_\q))$ is entailed from $\Pi''_{\omq},\A$, and hence is present in $\A$. 
Similarly, we can show that for every $R(v_\q,v_\q) \in \q$, the ground atom $R(\vec{c}(v_\q),\vec{c}(v_\q))$ belongs to $\A$. 
It follows that all of these atoms hold in $\can$ under $h$. Finally, we recall that $\vec{c}$ coincides with $\vec{a}$
on $\vec{z}$, so we have $h(\vec{z})=\vec{a}$, as required. 

\medskip

\item Case 2: $P_\q(\vec{a})$ was derived by an application of the following clause, for a tree witness $\t$ for $(\T,\q(\vec{z}))$ with $\tr\neq \emptyset$ and $v_\q\in\ti$ and role $R$ generating $\t$: 
\begin{equation*}
\rew_\q(\vec{z}) \ \ \ \leftarrow \bigwedge_{u,u' \in \tr} (u=u') \ \ \ \land \ \bigwedge_{u \in \tr} A_R(u)\  \
\land \bigwedge_{1\leq i \leq m} \rew_{\q_i^\t}(\vec{z}_i^\t),
\end{equation*}
where $\q_1^\t, \dots, \q_m^\t$ are the connected components of $\q$ without  $\q_\t$.
There must exist a substitution $\vec{c}$ for the variables in the body of this rule that coincides with $\vec{a}$ on $\vec{z}$ and is 
such that the ground atoms obtained by applying $\vec{c}$ to the variables in the body are all entailed from $\Pi''_{\omq},\A$. 
In particular, for every $1 \leq i \leq m$, we have $\Pi''_{\omq},\A \models \rew_{\q_i^\t}(\vec{c}(\vec{z}_i^\t))$. 
We can apply the induction hypothesis to the $\q_i^\t(\vec{z}_i^\t)$ to find homomorphisms $h_1, \ldots, h_m$ of $\q_1^\t, \ldots, \q_m^\t$ into $\can$
such that $h_i(\vec{z}_i^\t)=\vec{c}(\vec{z}_i^\t)$. 
Since $\t$ is a tree witness for $(\T,\q(\vec{z}))$ generated by $R$, there exists a homomorphism $h_\t$ of $\q_\t$ into 
$\C_\T^{\smash{A_R(a)}}$ with $\tr = h_\t^{-1}(a)$ and such that $h_\t(v)$ begins by $a R$ for every $v \in \ti$. 
Now pick some $u_0 \in \tr$ (recall that $\tr \neq \emptyset$). Then $A_R(u_0)$ is an atom in the clause body, and so $\Pi''_{\omq},\A \models A_R(\vec{c}(u_0))$,
which means that $A_R(\vec{c}(u_0))$ must appear in $\A$. It follows that for every element in $ \C_\T^{\smash{A_R(a)}}$
of the form $a R w$,  there exists a corresponding element $\vec{c}(u_0) R w$ in $\Delta^{\can}$. 
We now define a mapping $h$ from $\vars(\q)$ to $\Delta^{\can}$ as follows:
\begin{equation*}
h(v) = \begin{cases}
h_i(v), & \text{ for every } v \in \vars(\q_i^\t),\\
\vec{c}(u_0) R w, & \text{ if } v \in \ti \text{ and } h_\t(v)= a R w,\\
\vec{c}(u_0) & \text{ if } v \in \tr.
\end{cases}
\end{equation*}
Every variable in $\vars(\q)$ occurs  in $\tr\cup\ti$ or in exactly one of the $\q_i^\t$, and so is assigned a unique value by $h$. Note that although $\tr\cap\vars(\q_i^\t)$ is not necessarily empty, due to the equality atoms, we have $h(v) = h(v')$, for all $v,v'\in\tr$, and so the function is well-defined. 
We claim that $h$ is a homomorphism from $\q$ into $\can$. Clearly, the atoms occurring in some $\q_i^\t$ are preserved under $h$. Now consider some $A(v)$ with $v \in \ti$. Then $h(v)=\vec{c}(u_0) R w$, where $h_\t(v)= a R w$. 
Since $h_\t$ is a homomorphism, we know that $w$ ends with a role $S$ such that $\exists S^- \sqsubseteq_\T A$. 
It follows that $h(v)$ also ends with $S$, and thus $h(v) \in A^{\can}$. 
Next, consider a role atom $S(v,v')$, where at least one of $v$ and $v'$ belongs to $\ti$. 
As $h_\t$ is a homomorphism, either $h_\t(v') = h_\t(v) S'$ with $S' \sqsubseteq_\T S$,
or $h_\t(v)= h_\t(v') S'$ with $S' \sqsubseteq_\T S^-$, for some $S'$. 
We also know that $\vec{c}(u)=\vec{c}(u_0)$ for all $u \in \tr$, hence $h(u)=h(u_0)$ for all $u\in \tr$. 
It follows  that either $h(v') = h(v) S'$ with $S' \sqsubseteq_\T S$,
or $h(v)= h(v') S'$ with $S' \sqsubseteq_\T S^-$, and so $S(v,v')$ is preserved under  $h$. Finally, since $\vec{c}$ coincides with $\vec{a}$ on $\vec{z}$, we have $h(\vec{z})=\vec{a}$.  
\end{itemize}

For the converse direction of the induction step, suppose that $h$ is a homomorphism of $\q$ into $\can$ such that $h(\vec{z})=\vec{a}$. 
There are two cases to consider, depending on where $h$ maps the `splitting' variable $v_\q$.
\begin{itemize}
\item Case 1: $h(v_\q) \in \ind(\A)$. In this case, let $\q_1(\vec{z}_1), \ldots, \q_n(\vec{z}_n)$ be the subqueries of $\q(\vec{z})$ 
induced by the neighbours of $v_\q$ in $\gfmn$. Recall that $\vec{z}_i$ consists of $v_\q$ and the variables in $\vars(\q_i)\cap\vec{z}$. 
By restricting $h$ to $\vars(\q_i)$, we obtain, for each \mbox{$1 \leq i \leq n$}, a homomorphism of 
$\q_i(\vec{z}_i)$ into $\can$ that maps $v_\q$ to $h(v_\q)$ and $\vars(\q_i)\cap \vec{z}$ to $\vec{a}(\vars(\q_i)\cap\vec{z})$. Consider $\vec{a}^*$ defined by taking  $\vec{a}^*(z)= \vec{a}(z)$ for every $z \in  \vars(\q_i)\cap \vec{z}$
and $\vec{a}^*(v_\q)=h(v_\q)$. By the induction hypothesis, for every \mbox{$1 \leq i \leq n$}, we have
$\Pi''_{\omq},\A \models \rew_{\q_i}(\vec{a}^*(\vec{z}_i))$. Next, since $h$ is a homomorphism, we must have $h(v_\q) \in A^{\can}$ whenever $A(v_\q) \in \q$
and $(h(v_\q), h(v_\q)) \in R^{\can}$ whenever $R(v_\q,v_\q) \in \q$. Since $\A$ is H-complete,  $A(h(v_\q)) \in \A$
for every \mbox{$A(v_\q) \in \q$} and $R(h(v_\q),h(v_\q))$ for every $R(v_\q,v_\q) \in \q$. We have thus shown that, under the substitution $\vec{a}^*$,  every atom in the body of the clause
\begin{equation*}
\rew_\q(\vec{z}) \ \ \ \leftarrow \bigwedge_{A(v_\q)\in \q}\hspace*{-1em} A(v_\q) \ \ \ \land \bigwedge_{R(v_\q,v_\q) \in \q} \hspace*{-1em}R(v_\q,v_\q)  \ \ \  \land \bigwedge_{1\leq i\leq n} \rew_{\q_i}(\vec{z}_i)
\end{equation*}
is entailed from $\Pi''_{\omq},\A$. It follows that we must also have $\Pi''_{\omq},\A \models \rew_\q(\vec{a})$. 
\medskip

\item Case 2: $h(v_\q) \not \in \ind(\A)$. Then $h(v_\q)$ is of the form $b R w$. 
Let $V$ be the smallest subset of $\vars(\q)$ that contains $v_\q$
and satisfies the following closure property: 
\begin{itemize}
\item[--] if $v \in V$, 
$h(v) \notin \ind(\A)$ and $\q$ contains an atom with $v$ and $v'$, then $v' \in V$. 
\end{itemize}
Let $V'$ consist of all variables in $V$
such that $h(v) \not \in \ind(\A)$. We observe that $h(v)$ begins by $b R$ for every $v \in V'$ and $h(v)=b$
for every $v \in V \setminus V'$. 
Define $\q_V$ as the CQ comprising all atoms in $\q$ 
whose variables are in $V$ and which contain at least one variable from $V'$; the answer variables of $\q_V$ are $V\setminus V'$.
By replacing the initial $b$ by $a$ in the mapping $h$, we obtain a homomorphism $h_V$ of $\q_V$ into $ \C_\T^{\smash{A_R(a)}}$ with $V\setminus V' = h_V^{-1}(a)$. 
It follows that $\t=(\tr,\ti)$ with $\tr= V \setminus V'$ and $\ti =V'$ is a tree witness for $(\T,\q(\vec{z}))$ generated by $R$ (and $\q_\t = \q_V$). 
Moreover, $\tr\neq \emptyset$ because $\q$ has at least one answer variable. 
This means that the program $\Pi''_{\omq}$ contains the following clause
\begin{equation*}
\rew_\q(\vec{z}) \ \ \ \leftarrow \bigwedge_{u,u' \in \tr} (u=u') \ \ \ \land \ \bigwedge_{u \in \tr} A_R(u)\  \
\land \bigwedge_{1\leq i \leq m} \rew_{\q_i^\t}(\vec{z}_i^\t),
\end{equation*}
where $\q_1^\t, \dots, \q_m^\t$ are the connected components of $\q$ without $\q_\t $.
Recall that the query $\q_i^\t$ has answer variables $\vec{z}_i^\t = \vars(\q^\t_i)\cap (\vec{z} \cup \tr)$. 
Let $\vec{a}^*$ be the substitution for $\vec{z} \cup \tr$ such that $\vec{a}^*(z)=\vec{a}(z)$ for $z \in \vec{z}$
and $\vec{a}^*(v)=h(v)$ for $v \in \tr$. 
Then, for every $1 \leq i \leq m$, there exists a homomorphism $h_i$ from $\q_i^\t$ to $\can$ such that 
\mbox{$h_i(z) = \vec{a}^*(z)$} for every $z \in \vec{z}_i^\t$. By the induction 
hypothesis, \mbox{$\Pi''_{\omq}, \A \models  \rew_{\q_i^\t}(\vec{a}^*(\vec{z}_i^\t))$}. 
Next, since $h(v)=b$ for every $v \in \tr$, we have $\vec{a}^*(u)=\vec{a}^*(u')$ for every
$u,u' \in \tr$. Moreover, the presence of the element $b R$ in $\can$ means that 
$\T, \A \models A_R(b)$. Since $\A$ is H-complete, we have $A_R(b) \in \A$. 
It follows that under the substitution $\vec{a}^*$, all atoms in the body of the clause under consideration
are entailed by $\Pi''_{\omq}, \A $. Thus, we must also have $\Pi''_{\omq}, \A  \models \rew_\q(\vec{a})$. 
\end{itemize}

\smallskip

We have thus shown the lemma for all queries $\sqset$ other than $\q_0(\vec{x})$. 
Let us now turn to $\q_0(\vec{x})$.
For the first direction, suppose $\Pi''_{\omq},\A \models P_{\q_0}(\vec{a})$.
There are four cases, depending on which type of clause was used to derive $P_\q(\vec{a})$. 
We skip the first three cases, which are identical to those considered in the base case and induction step,
and focus instead on the case in which $P_{\q_0}(\vec{a})$ was derived using a clause of the form
$P_{\q_0} \leftarrow A(x)$ with $A$ a concept name such that $\T, \{A(a)\} \models \q_0$. 
In this case, there must exist some $b \in \ind(\A)$
such that $\T, \A \models A(b)$. By H-completeness of $\A$, we obtain $A(b) \in \A$. 
Since $\T, \{A(a)\} \models \q_0$, we get $\T, \A \models \q_0$, which implies the 
existence of a homomorphism from $\q_0$ into $\can$. 

For the converse direction, suppose that there is a homomorphism $h\colon \q_0 \to \can$ such that $h(\vec{x})=\vec{a}$. 
We focus on the case in which $\q_0$ is Boolean  ($\vec{x} = \emptyset$) and none of the variables in $\q_0$ is mapped to an ABox individual (the other cases
can be handled exactly as in the induction basis and induction step). In this case, there must 
exist an individual $b$ and role~$R$ such that $h(z)$ begins by $b R$ for every $z \in \vars(\q)$. It follows that $\T, \{A_R(a)\} \models \q_0$, since the mapping $h'$
defined by setting $h'(z)=a R w$ whenever $h(z) = b R w$ is a homomorphism from $\q$ to $\mathcal{C}_{\T, \{A_R(a)\}}$. 
It follows that $\Pi''_\omq$ contains the clause $P_{\q_0} \gets A_R(x)$. Since $b R$ occurs in $\Delta^{\can}$, we have $\T, \A \models A_R(b)$. 
By H-completeness of~$\A$, $A_R(b) \in \A$, and so by applying the clause $P_{\q_0} \gets A_R(x)$, we obtain $\Pi''_\omq,\A \models P_\q(\vec{a})$.\qed
\end{proof}

\section{Experiments}

\subsection{Computing rewritings}

We computed four types of rewritings for linear 
queries similar to those in Example~\ref{ex:rewriting:1}
and a fixed ontology from Example~\ref{ex:rewriting:2}.
We denote the rewriting from Section~\ref{sec:boundedtw}  by \textsc{Log} 
(because it is of logarithmic depth),  and from  
Section \ref{sec:5}  by \textsc{Lin} (because it is of linear depth). 
Other two rewritings were obtained by running executables of  Rapid~\cite{DBLP:conf/cade/ChortarasTS11} and Clipper~\cite{DBLP:conf/aaai/EiterOSTX12} with a 5 minute timeout on a desktop machine. 
We considered the following three sequences of letters $R$ and $S$:
\begin{align*}
\tag{Sequence 1}&RRSRSRSRRSRRSSR,\\
\tag{Sequence 2} &SRRRRRSRSRRRRRR,\\
\tag{Sequence 3} &SRRSSRSRSRRSRRSS.
\end{align*}
For each of the three sequences, we consider the line-shaped 
queries with 1--15 atoms formed by their prefixes. Table~\ref{tab:size} present the sizes of different types of rewritings.

\begin{table}[t]
\caption{The size (number of clauses) of different types of rewritings for the three sequences of queries ( -- indicates timeout after 5 minutes)}   \centering\tabcolsep=4pt
    \begin{tabular}{|c|r|r|r|r|r|r|r|r|r|r|r|r|}
    \hline
    no. & \multicolumn{4}{|c|}{Sequence~1} & \multicolumn{4}{|c|}{Sequence~2} & \multicolumn{4}{|c|}{Sequence~3} \\
    of & \multicolumn{4}{|c|}{\footnotesize\!$RRSRSRSRRSRRSSR$\!} & \multicolumn{4}{|c|}{\footnotesize\!$SRRRRRSRSRRRRRR$\!} & \multicolumn{4}{|c|}{\footnotesize\!$SRRSSRSRSRRSRRSS$\!} \\
    \cline{2-13}
    \!atoms\! & \!Rapid\! & \!\!Clipper\!\! & \textsc{Lin}   & \textsc{Log} & Rapid & \!\!Clipper\!\! & \textsc{Lin}   & \textsc{Log} & Rapid & \!\!Clipper\!\! & \textsc{Lin}   & \textsc{Log}    \\
    \hline
    1 & 1 & 1 & 2 & 1   & 1   & 1 & 2 & 1  & 1   & 1 & 2 & 1  \\
    2 & 1 & 1 & 5 & 2   & 2   & 2 & 5 & 4  & 2  & 2 & 5 & 4  \\
    3 & 2 & 2 & 8 & 5   & 2   & 2 & 8 & 5  & 2  & 2 & 8 & 5\\
    4 & 3  & 3 & 11    & 8 & 2   & 2 & 11    & 6     & 4  & 4 & 11    & 8\\
    5 & 5 & 5 & 14    & 12    & 2   & 2 & 14    & 8    & 4  & 4 & 14    & 10\\
    6 & 7  & 7 & 17    & 16   & 2  & 2 & 17    & 10   & 8  & 8 & 17    & 15\\
    7 & 10 & 11    & 20    & 20  & 4  & 4 & 20    & 13  & 11 & 11    & 20    & 18 \\
    8 & 13 & 16    & 23    & 24  & 6   & 7 & 23    & 16 & 18 & 24    & 23    & 21  \\
    9 & 13 & 16    & 26    & 27  & 10  & 13    & 26    & 22   & 24 & 35    & 26    & 27 \\
    10    & 26 & 44    & 29    & 32 & 14  & 26    & 29    & 27   & 34  & 63    & 29    & 33 \\
    11    & 39 & 72    & 32    & 36  & 14  & 26    & 32    & 29  & 43 & 100   & 32    & 37 \\
    12    & 39 & 126   & 35    & 40 & 14   & 26    & 35    & 33 & 56 & 302   & 35    & 42\\
    13    & -- & 241   & 38    & 45  & --   & 30    & 38    & 35   & -- & -- & 38    & 46\\
    14    & -- & -- & 41    & 47    & -- & 31    & 41    & 36    & -- & -- & 41    & 51 \\
    15    & -- & -- & 44    & 51 & --  & 30    & 44    & 37  & -- & -- & 44    & 52\\
    \hline
    \end{tabular}%
\label{tab:size}%
\end{table}%

\subsection{Datasets}\label{sec:datasets}

We used  Erd\"os-R\`enyi random graphs with independent 
parameters $V$ (number of vertices), $p$ (probability of an $R$-edge)
and $q$ (probability of concepts $A$ and $B$ at a given vertex).
Note that we intentionally did not introduce any $S$-edges.
The last parameter, the average degree of a vertex, 
is $V\cdot p$. Table~\ref{tab:datasets} summarises the parameters of the datasets.

\begin{table}[t]
  \centering\tabcolsep=6pt
  \caption{Generated datasets}
    \begin{tabular}{|c|r|r|r|r|r|}
    \hline
    dataset & \multicolumn{1}{|c|}{$V$} &\multicolumn{1}{|c|}{$p$} & \multicolumn{1}{|c|}{$q$} & {\tabcolsep=0pt\begin{tabular}{c}avg. degree\\[-2pt] of vertices\end{tabular}} & no.\ of atoms\\
    \hline
    4.ttl & 1\,000  & 0.050  & 0.050  & 50 &61\,498\\
    5.ttl & 5\,000  & 0.002 & 0.004 & 10 &64\,157\\
    6.ttl & 10\,000 & 0.002 & 0.004 & 20 &256\,804\\
    8.ttl & 20\,000 & 0.002 & 0.010  & 40 &1\,027\,028\\
    \hline
    \end{tabular}%
  \label{tab:datasets}%
\end{table}%

\nopagebreak[4]

\subsection{Evaluating rewritings}

We evaluated all obtained rewritings for the sequence  $RRSRSRSRRSRRSSR$
on the datasets in Section~\ref{sec:datasets} using RDFox triplestore~\cite{DBLP:conf/semweb/NenovPMHWB15}. The materialisation time and other relevant statistics are given in 
Table~\ref{tab:rdfox}.

\begin{table}[t!]
  \centering
  \caption{Evaluating rewritings on RDFox}
    \begin{tabular}{|c|r|r|r|r|r|r|r|r|r|r|}
    \hline
    data- & query  & \multicolumn{4}{c|}{evaluation time (sec)} & \multicolumn{1}{|c|}{no.\ of}
 & \multicolumn{4}{c|}{no.\ of generated tuples} \\\cline{3-6}\cline{8-11}
set   & \multicolumn{1}{|c|}{size} & \multicolumn{1}{|c|}{Rapid} & \multicolumn{1}{|c|}{Clipper} & \multicolumn{1}{|c|}{\textsc{Lin}} &  \multicolumn{1}{|c|}{\textsc{Log}} & \multicolumn{1}{|c|}{answers} & \multicolumn{1}{|c|}{Rapid} & \multicolumn{1}{|c|}{Clipper} & \multicolumn{1}{|c|}{\textsc{Lin}} & \multicolumn{1}{|c|}{\textsc{Log}} \\
\hline
      & 7   & 0.271     & 0.242     & 0.008      & 0.243      & 2\,956       & 2\,956      & 2\,956      & 3\,246        & 125\,361         \\
      & 8   & 0.412     & 0.377     & 0.084      & 0.904      & 212\,213     & 212\,213    & 212\,213    & 302\,221      & 1\,659\,409        \\
      & 9   & 3.117     & 3.337     & 3.376      & 2.941      & 998\,945     & 998\,945    & 998\,945    & 2\,927\,979   & 2\,684\,359        \\
4.ttl & 10  & 1.079     & 1.102     & 0.012      & 0.607      & 8\,374       & 8\,374      & 10\,760     & 12\,573       & 1\,178\,714        \\
      & 11  & 2.246     & 1.984     & 0.385      & 0.945      & 436\,000     & 436\,000    & 436\,000    & 836\,876      & 1\,618\,743        \\
      & 12  & 13.693    & 30.032    & 8.129      & 6.867      & 999\,998     & 999\,998    & 1\,000\,000 & 5\,311\,314   & 4\,439\,352        \\
      & 13  & --        & 6.810     & 0.027      & 0.616      & 20\,985      & --          & 24\,839     & 38\,200       & 553\,821         \\
      & 14  & --        & --        & 0.013      & 0.358      & 0            & --          & --          & 48        	 & 312\,723         \\
      & 15  & --        & --        & 0.032      & 0.394      & 2\,000       & --          & --          & 70\,277       & 376\,313         \\
\hline
      & 7   & 0.089     & 0.080     & 0.008      & 0.078      & 427          & 427         & 427         & 613           & 68\,546          \\
      & 8   & 0.136     & 0.125     & 0.029      & 0.434      & 8\,778       & 8\,778      & 8\,778      & 76\,202       & 1\,085\,362        \\
      & 9   & 0.202     & 0.254     & 0.369      & 0.554      & 105\,853     & 105\,853    & 105\,853    & 1\,020\,363   & 1\,190\,249        \\
5.ttl & 10  & 0.174     & 0.204     & 0.011      & 0.461      & 11           & 11          & 438         & 506           & 943\,097         \\
      & 11  & 0.192     & 0.259     & 0.036      & 0.473      & 651          & 651         & 9\,396      & 74\,922       & 944\,210         \\
      & 12  & 0.244     & 0.699     & 0.396      & 1.034      & 8\,058       & 8\,058      & 113\,179    & 1\,004\,735   & 1\,940\,300        \\
      & 13  & --        & 0.629     & 0.015      & 0.244      & 0            & --          & 438         & 502           & 209\,915         \\
      & 14  & --        & --        & 0.014      & 0.153      & 0            & --          & --          & 31            & 200\,962         \\
      & 15  & --        & --        & 0.032      & 0.172      & 0            & --          & --          & 64\,543       & 265\,087         \\
\hline
      & 7   & 0.631     & 0.581     & 0.035      & 0.756      & 1\,217       & 1\,217      & 1\,217      & 1\,499        & 296\,711         \\
      & 8   & 0.925     & 0.876     & 0.159      & 4.377      & 67\,022      & 67\,022     & 67\,022     & 335\,578      & 7\,546\,184        \\
      & 9   & 1.949     & 2.275     & 4.063      & 5.251      & 1\,678\,668  & 1\,678\,668 & 1\,678\,668 & 8\,613\,829   & 9\,225\,201        \\
6.ttl & 10  & 1.24      & 1.377     & 0.049      & 4.731      & 60           & 60          & 1\,277      & 1\,389        & 6\,936\,178        \\
      & 11  & 1.403     & 1.798     & 0.249      & 4.846      & 11\,498      & 11\,498     & 77\,811     & 341\,459      & 6\,949\,160        \\
      & 12  & 1.697     & 5.413     & 4.355      & 10.128     & 305\,640     & 305\,640    & 1\,951\,654 & 8\,780\,232   & 15\,626\,926       \\
      & 13  & --        & 4.382     & 0.082      & 1.762      & 0            & --    	   & 1\,277      & 1\,377     	 & 917\,117         \\
      & 14  & --        & --        & 0.063      & 1.115      & 0            & --          & --          & 47      	     & 850\,309         \\
      & 15  & --        & --        & 0.177      & 1.011      & 0            & --    	   & --          & 257\,974      & 1\,107\,065        \\
\hline
      & 7   & 6.614     & 6.277     & 0.243      & 8.586      & 13\,103      & 13\,103     & 13\,103     & 14\,625        & 1\,665\,376        \\
      & 8   & 11.441    & 10.923    & 1.880      & 54.813     & 1\,286\,991  & 1\,286\,991 & 1\,286\,991 & 2\,432\,629    & 56\,098\,445       \\
      & 9   & 46.704    & 50.668    & 76.169     & 102.055    & 58\,753\,514 & 58\,753\,514& 58\,753\,514& 114\,973\,160  & 114\,837\,395      \\
8.ttl & 10  & 14.348    & 15.503    & 0.375      & 43.347     & 19\,966      & 19\,966     & 33\,014     & 35\,359        & 52\,103\,362       \\
      & 11  & 19.593    & 20.907    & 2.843      & 44.410     & 1\,872\,159  & 1\,872\,159 & 3\,051\,184 & 4\,397\,556    & 53\,986\,724       \\
      & 12  & 71.354    & 182.499   & 172.822    & 237.478    & 79\,939\,048 & 79\,939\,048& 120\,229\,590 & 199\,083\,489& 242\,500\,074      \\
      & 13  & --        & 54.497    & 0.562      & 22.345     & 22\,474      & --          & 53\,717     & 58\,826        & 5\,686\,759        \\
      & 14  & --        & --        & 0.550      & 12.462     & 0            & --          & --          & 253            & 4\,356\,739        \\
      & 15  & --        & --        & 1.211      & 11.315     & 12\,165      & --          & --          & 1\,064\,542    & 5\,395\,902        \\
    \hline
    \end{tabular}%
  \label{tab:rdfox}%
\end{table}%

\clearpage

\bibliographystyle{splncs03}

\begin{thebibliography}{10}
\providecommand{\url}[1]{\texttt{#1}}
\providecommand{\urlprefix}{URL }

\bibitem{Abitebouletal95}
Abiteboul, S., Hull, R., Vianu, V.: Foundations of Databases. Addison-Wesley
  (1995)

\bibitem{BCMNP03}
Baader, F., Calvanese, D., McGuinness, D., Nardi, D., Patel-Schneider, P.
  (eds.): The Description Logic Handbook: {T}heory, Implementation and
  Applications. Cambridge University Press (2003)

\bibitem{DBLP:conf/lics/BienvenuKP15}
Bienvenu, M., Kikot, S., Podolskii, V.V.: Tree-like queries in {OWL} 2 {QL:}
  succinctness and complexity results. In: Proc.\ of the 30th Annual {ACM/IEEE}
  Symposium on Logic in Computer Science (LICS 2015). pp. 317--328. ACM (2015)

\bibitem{DBLP:journals/ws/CaliGL12}
Cal\`{\i}, A., Gottlob, G., Lukasiewicz, T.: A general datalog-based framework
  for tractable query answering over ontologies. Journal of Web Semantics  14,
  57--83 (2012)

\bibitem{CDLLR07}
Calvanese, D., De~Giacomo, G., Lembo, D., Lenzerini, M., Rosati, R.: Tractable
  reasoning and efficient query answering in description logics: the
  {{\textit{DL-Lite}}} family. Journal of Automated Reasoning  39(3),  385--429
  (2007)

\bibitem{DBLP:conf/cade/ChortarasTS11}
Chortaras, A., Trivela, D., Stamou, G.: Optimized query rewriting for {OWL 2
  QL}. In: Proc.\ of CADE-23. LNCS, vol. 6803, pp. 192--206. Springer (2011)

\bibitem{DBLP:journals/jacm/Cook71}
Cook, S.A.: Characterizations of pushdown machines in terms of time-bounded
  computers. Journal of the {ACM}  18(1),  4--18 (1971)

\bibitem{DBLP:conf/aaai/EiterOSTX12}
Eiter, T., Ortiz, M., {\v{S}}imkus, M., Tran, T.K., Xiao, G.: Query rewriting
  for {Horn-SHIQ} plus rules. In: Proc.\ of the 26th AAAI Conf.\ on Artificial
  Intelligence (AAAI 2012). pp. 726--733. AAAI (2012)

\bibitem{DBLP:journals/ai/GottlobKKPSZ14}
Gottlob, G., Kikot, S., Kontchakov, R., Podolskii, V.V., Schwentick, T.,
  Zakharyaschev, M.: The price of query rewriting in ontology-based data
  access. Artificial Intelligence  213,  42--59 (2014)

\bibitem{DBLP:conf/icalp/GottlobLS99}
Gottlob, G., Leone, N., Scarcello, F.: Computing {LOGCFL} certificates. In:
  Proc. of the 26th Int. Colloquium on Automata, Languages and Programming
  (ICALP-99). LNCS, vol. 1644, pp. 361--371. Springer (1999)

\bibitem{huf52}
Huffman, D.A.: A method for the construction of minimum-redundancy codes.
  Proc.\ of the Institute of Radio Engineers  40(9),  1098--1101 (1952)

\bibitem{LICS14}
Kikot, S., Kontchakov, R., Podolskii, V., Zakharyaschev, M.: On the
  succinctness of query rewriting over shallow ontologies. In: Proc.\ of the
  29th Annual ACM/IEEE Symposium on Logic in Computer Science (LICS 2014). ACM
  (2014)

\bibitem{DBLP:conf/icalp/KikotKPZ12}
Kikot, S., Kontchakov, R., Podolskii, V.V., Zakharyaschev, M.: Exponential
  lower bounds and separation for query rewriting. In: Proc.\ of the 39th Int.\
  Colloquium on Automata, Languages and Programming (ICALP 2012). LNCS, vol.
  7392, pp. 263--274. Springer (2012)

\bibitem{DBLP:conf/dlog/KikotKZ11}
Kikot, S., Kontchakov, R., Zakharyaschev, M.: On (in)tractability of {OBDA}
  with {OWL 2 QL}. In: Proc.\ of the 24th Int.\ Workshop on Description Logics
  (DL 2011). vol. 745, pp. 224--234. CEUR-WS (2011)

\bibitem{DBLP:conf/kr/KikotKZ12}
Kikot, S., Kontchakov, R., Zakharyaschev, M.: Conjunctive query answering with
  {OWL~2~QL}. In: Proc.\ of the 13th Int.\ Conf.\ on Principles of Knowledge
  Representation and Reasoning (KR 2012). pp. 275--285. AAAI (2012)

\bibitem{KR10our}
Kontchakov, R., Lutz, C., Toman, D., Wolter, F., Zakharyaschev, M.: The
  combined approach to query answering in {DL-Lite}. In: Proc.\ of the 12th
  Int.\ Conf.\ on Principles of Knowledge Representation and Reasoning (KR
  2010). AAAI Press (2010)

\bibitem{DBLP:conf/semweb/NenovPMHWB15}
Nenov, Y., Piro, R., Motik, B., Horrocks, I., Wu, Z., Banerjee, J.: {RDFox:}
  {A} highly-scalable {RDF} store. In: Proc.\ of the 14th Int.\ Semantic Web
  Conf.\ (ISWC 2015), Part {II}. LNCS, vol. 9367, pp. 3--20. Springer (2015)

\bibitem{PLCD*08}
Poggi, A., Lembo, D., Calvanese, D., De~Giacomo, G., Lenzerini, M., Rosati, R.:
  Linking data to ontologies. Journal on Data Semantics  X,  133--173 (2008)

\bibitem{sudborough78}
Sudborough, I.H.: On the tape complexity of deterministic context-free
  languages. Journal of the {ACM}  25(3),  405--414 (1978)

\bibitem{DBLP:journals/jcss/Venkateswaran91}
Venkateswaran, H.: Properties that characterize {LOGCFL}. Journal of Computer
  and System Sciences  43(2),  380--404 (1991)

\end{thebibliography}

\end{document}